\documentclass[twopages,11pt]{article}
\usepackage{amsmath}
\usepackage{amsthm}
\usepackage{amssymb}
\usepackage{amsfonts}
\usepackage{mathrsfs}
\usepackage{amscd}
\usepackage[british]{babel}
\usepackage{graphicx,subfigure}
\usepackage{times}
\usepackage{url}
\usepackage{psfrag}
\usepackage{color}

\def\sform{\mbox{\boldmath $\omega$}}
\newcommand{\Dif}{\mathrm{D}}
\newcommand{\Lie}[1]{\mathrm{L}_{#1}}
\newcommand{\ep}{\varepsilon}

\newcommand{\RR}{\mathbb R}

\newcommand{\TT}{\mathbb T}
\newcommand{\cI}{\mathcal I}
\newcommand{\cJ}{\mathcal J}
\newcommand{\cT}{\mathcal T}

\newcommand{\cD}{\mathcal D}
\newcommand{\cF}{\mathcal F}
\newcommand{\cP}{\mathcal P}
\newcommand{\cL}{\mathcal L}

\newcommand{\cO}{\mathcal O}
\newcommand{\cR}{\mathcal R}
\newcommand{\cZ}{\mathcal Z}
\newcommand{\ZZ}{\mathbb Z}

\newcommand{\rme}{\mathrm{e}}

\newcommand{\srm}{\mathrm{s}}
\newcommand{\urm}{\mathrm{u}}
\DeclareMathOperator\curl{curl}

\theoremstyle{plain}
\newtheorem{definition}{Definition}[section]
\newtheorem{theorem}[definition]{Theorem}
\newtheorem{lemma}[definition]{Lemma}

\newtheorem{proposition}[definition]{Proposition}
\newtheorem{corollary}[definition]{Corollary}

\newtheorem{remark}[definition]{Remark}

\title{Arnold diffusion of charged particles in ABC magnetic fields}

\author{
Alejandro Luque\thanks{luque@icmat.es}\\ 
Instituto de Ciencias Matem\'aticas\\
Consejo Superior de Investigaciones Cient\'{\i}ficas \\
28049 Madrid (Spain).
\and 
Daniel Peralta-Salas\thanks{dperalta@icmat.es}\\ 
Instituto de Ciencias Matem\'aticas\\
Consejo Superior de Investigaciones Cient\'{\i}ficas \\
28049 Madrid (Spain).
}
\begin{document}
\maketitle

\thispagestyle{empty}

\begin{abstract}
We prove the existence of diffusing solutions in the motion of a charged particle in the
presence of an ABC magnetic field. The equations of motion are modeled by a 3DOF Hamiltonian
system depending on two parameters. For small values of these parameters, we obtain
a normally hyperbolic invariant manifold and we apply the so-called geometric
methods for a priori unstable systems developed by A.~Delshams, R.~de la Llave, and T.M.~Seara.
We characterize explicitly sufficient conditions for the existence of a transition
chain of invariant tori having heteroclinic connections, thus obtaining global instability (Arnold diffusion). We also check the obtained conditions in a computer assisted proof. ABC magnetic fields are the simplest force-free type solutions of the magnetohydrodynamics equations with periodic boundary conditions, so our results are of potential interest in the study of the motion of plasma charged particles in a tokamak. 

\bigskip 



\end{abstract}

\newpage


\markboth{}{A. Luque and D. Peralta-Salas} 
\section{Introduction}\label{S:intro}

The study of the motion of a charged particle in a magnetic field
is a classical subject
in several areas of physics, such as condensed matter
theory, accelerator physics, magnetobiology and plasma physics. 
The equation of motion of a (non-relativistic) unit-mass, unit-charge particle at the position $q\in \RR^3$ in the presence of a magnetic field $\bf B$ is given by the Newton-Lorentz law
\begin{equation}\label{firsteq}
\ddot{q} = \dot{q} \times {\bf B}(q)\,,
\end{equation}
where the dot over $q$ denotes, as usual, the time derivative, and $\times$ stands for
the standard vector product in $\RR^3$.

An important observation
is that Eq.~\eqref{firsteq} can be written equivalently in a Hamiltonian way whenever there is a globally defined vector potential $\bf A$ such that ${\bf B}= \curl\, {\bf A}$. If this is the case, the Hamiltonian function is
\begin{equation*}
H(q,p)=\frac{1}{2} (p-{\bf A}(q))^2\,.
\end{equation*}

In this paper we are interested in the motion of charges in ABC
magnetic fields. These fields arise in the theory of magnetic dynamos
(see~\cite{Gilbert} and references therein) and were
introduced independently by Arnold~\cite{Ar65} and
Childress~\cite{Ch70} in the 1960's. The well-known family of ABC magnetic fields 
depends on three real parameters, $A$, $B$ and $C$, and reads in Cartesian coordinates $q=(x,y,z)$ as
\begin{equation}\label{eq:mag:B:abc}
{\bf B}_{ABC} = (A \sin z + C \cos y, B \sin x + A \cos z, C \sin y + B \cos x)\,.
\end{equation}

ABC magnetic fields are stationary solutions of the magnetohydrodynamics
equations of force-free type, 
thus implying that the field exerts no force
on the current distribution
generating it. Indeed, it is straightforward to check that ${\bf B}_{ABC}$ is divergence-free
and force-free because $\curl {\bf B}_{ABC}={\bf B}_{ABC}$, and so the ABC field
admits the globally defined vector potential ${\bf A}_{ABC} = 
{\bf B}_{ABC}$. ABC magnetic fields are minimizers of the energy functional $\int \bf B^2$ acting on the space of divergence-free fields of fixed helicity.
 
Since the dependence of the ABC magnetic field and its vector
potential with the variables $x,y,z$ is $2\pi$-periodic, it is
customary to consider that these fields are defined in the $3$-torus
$\TT^3=\RR^3 /(2\pi \ZZ)^3$ so that $(x,y,z)\in\TT^3$.
By rescaling and reordering the space variables and the time,
all the non-trivial cases can be reduced to $A=1 \geq B \geq C \geq 0$, so we shall assume it in what follows.
The Newton-Lorentz equation of motion~\eqref{firsteq} for the ABC magnetic field can be described as a 3DOF Hamiltonian system defined in the phase space $\TT^3\times \RR^3 \ni (x,y,z,p_x,p_y,p_z)$ by the Hamiltonian function:
\begin{equation}\label{eq:habc}
H=\frac12(p_x-C\cos y-\sin z)^2+\frac12(p_y-B\sin x-\cos z)^2+\frac12(p_z-B\cos x-C\sin y)^2\, .
\end{equation}

Force free fields are very important in
applications and model diverse physical systems, as stellar
atmospheres~\cite{ChWo58}, the solar corona~\cite{Fl04} and relaxed
states of toroidal 
plasmas~\cite{Ta86}.
Moreover, the motion of a charge in an ABC field can be interpreted as
a model for the motion of plasma charged particles in a tokamak.
A wide examination of system~\eqref{eq:habc} was recently presented in~\cite{LP13},
proving the existence of confinement regions of charges near some magnetic lines
and also that the problem gives rise to non-integrability and chaotic motions.
In this study we go one step further and obtain global instability, i.e. Arnold diffusion.

Characterizing global instabilities in Hamiltonian systems
is a relevant problem that has called the attention of mathematicians,
physicists and engineers. For example, in the context of beam physics,
designers of accelerators or plasma confinement devices
are interested in the characterization of these instabilities in
order to avoid them as much as possible (e.g. in the confinement of hot plasmas for fusion power generation, diffusion is a very relevant phenomenon because of the harmful plasma-wall interaction).  Global instability deals
with the question of whether Hamiltonian perturbations of a regular integrable
system accumulate over time, giving rise to a long term effect,
or whether they average out. This problem was first formulated
by Arnold. Indeed, in the celebrated paper~\cite{Arnold64},
Arnold constructed a concrete example, suitably and cleverly chosen,
such that some trajectories can jump around KAM tori thus obtaining
diffusion (after~\cite{Arnold64} this problem is known as \emph{Arnold diffusion}).
These diffusing orbits were constructed using a mechanism of
transition chains. It consists in obtaining heteroclinic intersections between
the stable and unstable manifolds of a sequence of whiskered invariant tori.

In the last decades there has been a significant advance in the understanding of diffusion
and, following~\cite{CG94}, the studies are classified in two different groups:
the \emph{a priori unstable} case and the \emph{a priori stable case}.
Arnold diffusion in a priori unstable systems (where the unperturbed system has hyperbolic properties
of some kind) has been approached
using geometric methods in~\cite{DLS06, DLS13, DH09, GT08}, the separatrix map in~\cite{Tre04,Tre12},
topological methods in~\cite{GL06,GR13} and variational methods in~\cite{Bernard08,CY04}.
A combination of topological and geometric methods has been 
recently
presented in~\cite{GLS14}. The more difficult case of
a priori stable systems (where the unperturbed system is foliated
by Lagrangian invariant tori) is less understood, but significant
advances have been presented 
along the last few years
in~\cite{BKZ11,
Che13,
KS12,
KZ12,
KZ14a,
KZ14b,
Mat03,
Zha11}.

Our aim in this paper is to prove the existence of Arnold diffusion in the dynamics of a charged particle in an ABC magnetic field, which is modeled by the Hamiltonian system~\eqref{eq:habc}.
If $B=C=0$, we obtain an integrable Hamiltonian system $H_0$ having a
normally hyperbolic invariant manifold (NHIM) $\Lambda_0$ foliated by
whiskered invariant tori (see
details in Section~\ref{ssec:geom:unper}).
Then, the problem considered in this paper falls into the a priori unstable setting.
It is worth mentioning that one of the main difficulties in the study of a priori
unstable systems was the so-called \emph{large gap problem} (see~\cite{Moeckel96}).
This problem arises because a generic perturbation of size $\ep$ creates gaps at most of size $\sqrt{\ep}$
between the persisting primary KAM tori and, in principle,
only orbits separated an amount $\ep$ could be connected by heteroclinic connections between
invariant tori. 
This issue was solved in the previously mentioned
references, using different tools for the study of Arnold diffusion.
We observe that recent mechanisms of diffusion have
been proposed in order to avoid big gaps using very little
information of the dynamics restricted to the NHIM (see~\cite{CapinskiGL,DKRS}).
Here, we follow the geometric methods developed in~\cite{DLS06,DLS13} in order to
prove the existence of Arnold diffusion in the Hamiltonian~\eqref{eq:habc} for small
values of $B$ and $C$.
Concretely, we prove the following theorem, which establishes sufficient conditions
for the existence of a transition chain between 
whiskered
invariant tori, thus producing
large unstable motions in the perturbed system:
\\

\noindent {\bf Main Theorem (informal statement).} 
\emph{Let us consider the Hamiltonian~\eqref{eq:habc}
with $B=\ep \hat B \neq 0$ and $C=\ep \hat C \neq 0$,
and a non-empty set $\cI=[a_1,b_1] \times [a_2, b_2]$ for given
(positive) values of $a_i,b_i$.
Then, under some explicit non-degeneracy and transversality conditions,
if $|\ep|$ is small enough, the ABC system exhibits Arnold diffusion in $\cI$, i.e.
there exists a trajectory of~\eqref{eq:habc}
connecting two arbitrary values of $(p_x,p_y)$ in 
the interior of $\cI$.}
\\

A precise statement of this theorem is given in
Theorem~\ref{teo:diffusion:ABC} (Section~\ref{sec:setting}), after
a detailed discussion of the unperturbed ABC system.
Moreover, we implement
the  non-degeneracy and transversality conditions
included in the Main Theorem
in a computer assisted proof (CAP) in Section~\ref{sec:explicit}.
As a consequence, we obtain an open set of initial conditions
in phase space where we can construct a transition chain. 
For example,
we obtain the following result which serves as an illustration:

\begin{corollary}\label{cor:diffusion:ABC:informal}
Let us consider 
Hamiltonian~\eqref{eq:habc}
with 
$\hat B= 10$ and $\hat C= 0.1$. 
Then, the non-degeneracy and transversality conditions of the Main Theorem
hold in the set $\cI = [0.1, 0.9]\times [0.5, 0.9]$.
Therefore, for $|\ep|$ small enough,
there exists a trajectory of~\eqref{eq:habc}
connecting two arbitrary values of $(p_x,p_y)$ in
$(0.1, 0.9)\times (0.5, 0.9)$.
\end{corollary}

We remark that the choice $\hat B= 10$ and $\hat C=0.1$
is arbitrary. Analogous results can be obtained for any
other choice of parameters. The computational cost to verify the
hypotheses for a fixed set $\cI$ increases when the difference
between $\hat B$ and $\hat C$ is reduced. It is worth mentioning
that if we take ``narrow'' sets of the form
$\cI=[a_1,a_1+\delta] \times [a_2,b_2]$
or
$\cI=[a_1,b_1] \times [a_2,a_2+\delta]$,
with $\delta$ small, then the computational cost of the CAP
is reduced significantly. In this case, it is also possible to check
the conditions for open sets of parameters $\hat B$ and $\hat C$.
We have produced analogous results
to Corollary~\ref{cor:diffusion:ABC:informal} and we have
not found obstructions to diffusion in any case.

To the best of our knowledge, the Main Theorem
and Corollary~\ref{cor:diffusion:ABC:informal}
are the first rigorous results on the existence of diffusing orbits
in the motion of charges in magnetic fields, even though physicists
have been aware of this phenomenon for a long time 
(cf.~\cite{WBM92,ZZSUC86}) and the effect is sometimes known as drift motion
in the physics literature.
Of course, we want to mention other significant problems
where Arnold diffusion have been characterized. In particular, we can find
remarkable contributions in the context
of celestial mechanics:
diffusion along 
mean motion resonances
in the restricted planar three-body problem~\cite{FGKR13};
instability mechanism in a special configuration of the 5-body problem~\cite{Moeckel96,Zhe10};
transition chains of invariant tori around the point $L_2$ in the elliptic three body problem
as a perturbation of the circular problem~\cite{CZ11}, improved recently in~\cite{CapinskiGL};
instability around the point $L_1$ in the circular spatial restricted three-body problem,
focusing on
homoclinic trajectories~\cite{DGR};
instability in the elliptic restricted problem close to
the parabolic orbits of the Kepler problem between the comet and the Sun~\cite{DKRS}.
We observe that some parts of the arguments in~\cite{CapinskiGL,CZ11,DGR,FGKR13}
are non-rigorous, but are strongly backed by convincing numerical computations.
It is also worth mentioning the example discussed in~\cite{DH11},
where the geometric mechanism for diffusion introduced in~\cite{DH09}
is illustrated in a representative model. 
The model simplifies some of the hypotheses,
thus saving a significant amount of computations, so
they can present the geometric mechanism of diffusion in a clear understandable way.
In the system~\eqref{eq:habc} studied in this paper,
some of these simplifications cannot be
used and we must 
perform some \emph{ad hoc} analysis and specific computations.
The reader interested in numerical studies is referred
to~\cite{GuzzoLF09}.

The mechanisms governing Arnold diffusion are very complex and there 
are still many questions to answer and many aspects to understand. As is
posed in~\cite{GR13},
it is  relevant  to  detect,  combine,  and  compare  different  mechanisms
of  diffusion  displayed  by  concrete  systems.
In this way,
Hamiltonian~\eqref{eq:habc}
can be an ideal framework to apply and compare different approaches
and methods in the literature (e.g. topological methods, variational techniques,
use of multiple scattering maps, etc). On the one hand, the ABC system is complicated
enough to contain all the difficulties
that are present in a general a priori unstable problem. 
On the other hand, the ABC system is explicit and simple enough to perform
analytic computations. Moreover, it is a problem that appears in a natural
way in physics.

The proof of the Main Theorem consists
in combining the internal dynamics on the NHIM with its outer (asymptotic)
dynamics, which is modeled by the \emph{scattering map}~\cite{DLS08}.
The procedure is divided in the following steps:

\begin{description}
\item [Characterization of the NHIM:]
The first step is to characterize the
perturbed NHIM $\Lambda_\ep$ and its stable
and unstable manifolds
(we summarize some basic concepts in Section~\ref{ssec:inner:approx}).
We pay special attention to describe explicitly
the geometric procedure that allows us to parameterize the NHIM
in a natural way, thus obtaining a suitable symplectic structure
on the NHIM (see Section~\ref{ssec:NHIM:symp}). 
The construction presented has special interest since
we give explicit formulas to use
the deformation theory introduced in~\cite{DLS08}.
To this end, we have to compute perturbatively a symplectic frame
associated to the manifold.
Explicit computations for the ABC system are detailed in Section~\ref{ssec:NHIM:ABC}.

\item [Invariant tori on the NHIM:]
To study the inner dynamics on the NHIM, where the so-called \emph{big gaps}
are present, we perform averaging theory of the vector field restricted
to the manifold. After choosing a suitable parameterization
in the previous step, we follow~\cite{DLS06,DLS13} \emph{mutatis mutandis}
in Section~\ref{sssec:aver}.
Explicit computations for the ABC system are detailed in Section~\ref{ssec:aver:ABC}.
In Proposition~\ref{prop:level:sets} we obtain an approximation of the level sets that characterize
the invariant objects inside the NHIM. 
In particular, we find a set of whiskered invariant tori (primary and secondary)
covering $\Lambda_\ep$ except for a set of measure $\cO(\ep^{3/2})$.

\item [Scattering map:] In Section~\ref{sec:outer} we describe the outer dynamics associated to our
problem. For the sake of completeness, in Section~\ref{ssec:poincare}
we summarize the construction of the Melnikov
potential that characterizes the intersections of the stable and unstable manifolds associated to the NHIM
(cf.~\cite{Tre02}).
In Section~\ref{ssec:scattering} we compute the \emph{scattering map} for the ABC
system.

\item [Combination of inner and outer dynamics:] The combination of both dynamics, obtaining explicit
transversality conditions for the existence of diffusion, is performed
in Section~\ref{sec:chains}. We remark that, since the unperturbed scattering map 
has a so-called phase shift, there is an additional term in the transversality conditions that is not present
in~\cite{DLS06,DLS13}.
In the domain where the conditions are satisfied, we construct a sequence
$\{\cT_i\}_{i=1}^\infty$ of whiskered tori  satisfying
$W^{\urm}_{\cT_i} \pitchfork W^{\srm}_{\cT_{i+1}}$,
that is, we construct a transition chain along $\Lambda_\ep$.
\end{description}

We remark again that the hypotheses in
the Main Theorem are explicit
and involve a series of standard,
but cumbersome, computations. 
First, we evaluate some integrals that depend on $(p_x,p_y)$
as parameters.
We solve
a one-dimensional non-linear equation that depends on
these integrals. 
We approximate the
derivatives with respect to parameters of the previous solution.
Finally, we evaluate several complicated formulas that depend
on the previous objects.
In Section~\ref{sec:explicit} we rigorously perform
these computations with the help of a computer.

\section{Setting of the problem and statement of the main theorem}\label{sec:setting}

In this Section we present a detailed description of the 
geometry of our problem and state a precise version of
the Main Theorem. More precisely, in Section~\ref{ssec:geom:unper} we
fully describe the motion of the unperturbed Hamiltonian system (Eq.~\eqref{eq:habc} with $B=C=0$),
thus characterizing a normally hyperbolic invariant manifold
with coincident stable and unstable invariant manifolds.
Then, in Section~\ref{ssec:mainresult} we provide explicit
sufficient conditions for the existence of Arnold
diffusion in the perturbed problem (Eq.~\eqref{eq:habc} with $B=\ep\hat B$, $C=\ep \hat C$).

\subsection{Geometric features of the unperturbed
problem}\label{ssec:geom:unper}

For $B=C=0$, the ABC magnetic field has the simple expression
$${\bf B}_{ABC}=(\sin z,\cos z,0)\,,$$
which implies that the field is linear on each toroidal surface
${z=z_0}$, periodic or quasi-periodic depending on the value of $\tan
z_0$. Concerning the equations of motion, the Hamiltonian function in
Eq.~\eqref{eq:habc} is given by
\begin{equation}\label{eq:sys:H0}
H_0=\frac12(p_x-\sin z)^2+\frac12(p_y-\cos z)^2+\frac12 p_z^2\,.
\end{equation}
The system of ODEs associated to~\eqref{eq:sys:H0} is
\[
\begin{array}{lll}
\dot x =p_x-\sin z\,, &\qquad & \dot p_x =0\,,\\
\dot y =p_y-\cos z\,, &\qquad & \dot p_y =0\,,\\
\dot z =p_z\,,        &\qquad & \dot p_z = p_x \cos z - p_y \sin z\,,
\end{array}
\]
so $p_x$ and $p_y$ are constants of the motion. There is no loss of generality in taking positive values of $p_x$ and $p_y$, so we shall assume it throughout the paper. In addition, we observe that the system $(z,p_z)$ is pendulum-like and has an effective potential
\[
V(z):=-p_x \sin z - p_y \cos z\,.
\]
Notice that this system
has a hyperbolic equilibrium at the point
\[
z^*:=\arctan \frac{p_x}{p_y}+\pi, \qquad p_z^*=0\,,
\]
and, since $p_x>0$ and $p_y>0$, we have the identities
\[
\sin z^* = \frac{-p_x}{\sqrt{p_x^2+p_y^2}}\,, \qquad 
\cos z^* = \frac{-p_y}{\sqrt{p_x^2+p_y^2}}\,.
\]
We denote the positive eigenvalue of the linearized equation at the hyperbolic equilibrium as
\begin{equation}\label{eq:lambda}
\lambda:=(p_x^2+p_y^2)^{1/4}\,,
\end{equation}
which allows us to write the constants of the motion as 
$p_x=\lambda^2 \sin \alpha$,
and
$p_y = \lambda^2 \cos \alpha$,
with $\alpha=\arctan (p_x/p_y)\in (0,\pi/2)$. With this notation, the pendulum-like equation in the variables $(z,p_z)$ reads as
\[
\ddot z= p_x \cos z - p_y \sin z = \lambda^2 \sin (\alpha-z)\,,
\]
thus obtaining that there is a homoclinic orbit connecting the equilibrium point given by
\begin{equation}\label{eqsep}
z^0(t)=4 \arctan \rme^{\lambda t} + z^*\,, \qquad
p^0_z(t)= \frac{2\lambda}{\cosh (\lambda t)}\,.
\end{equation}
It is straightforward to check that $z^0(t) \rightarrow z^*$ and $p_z^0(t)\rightarrow p_z^*=0$, exponentially with exponent $\lambda$, as $t\to\pm\infty$. There is a second homoclinic trajectory connecting the equilibrium point given by $\bar z^0(t)=-z^0(t)+2\alpha$ and $\bar p_z^0(t)=-p^0_z(t)$, but it will not be used in what follows.

The previous computations show that the Hamiltonian system $H_0$ has a $4$-dimensional normally hyperbolic invariant manifold
\[
\Lambda_0 := \{ (q,p) \in \TT^3 \times \RR^3\,:\, z=z^*, p_z=p_z^*\}\,, 
\]
which is foliated by $2$-dimensional invariant tori $\cT_{p_x,p_y}$ obtained
by fixing $p_x$ and $p_y$, i.e.
$\Lambda_0 = \bigcup_{p_x,p_y} \cT_{p_x,p_y}$. A direct computation shows that the dynamics on each invariant torus $\cT_{p_x,p_y}$ is linear with frequency vector $\omega=(\omega_1,\omega_2)$ given by
\begin{align}
\omega_1 := {} & p_x -\sin(z^*) = p_x(1+(p_x^2+p_y^2)^{-1/2}) \,, \label{eq:omega1}\\
\omega_2 := {} & p_y -\cos(z^*) = p_y(1+(p_x^2+p_y^2)^{-1/2}) \,. \label{eq:omega2}
\end{align}
The stable and unstable manifolds of $\Lambda_0$
are $5$-dimensional invariant sets defined by
\[
W^\srm(\Lambda_0)= W^\urm(\Lambda_0)=\{ (q,p) \in \TT^3 \times \RR^3\,:\,
z=z^0(\tau),
p_z=p_z^0(\tau), \tau\in \RR\}\,,
\]
so 
the set $W^\srm(\Lambda_0)$ (or $W^\urm(\Lambda_0)$) is the union of the stable (unstable) manifolds of the invariant tori $\cT_{p_x,p_y}$, i.e.
\[
W^\srm(\Lambda_0) = \bigcup_{p_x,p_y} W^\srm(\cT_{p_x,p_y})= W^\urm(\Lambda_0) = \bigcup_{p_x,p_y} W^\urm(\cT_{p_x,p_y})\,.
\]

In order to work with the invariant torus $\cT_{p_x,p_y}$ and its
whiskers $W^\srm(\cT_{p_x,p_y})=W^\urm(\cT_{p_x,p_y})$, we introduce
appropriate parameterizations. Indeed,
$\cT_{p_x,p_y}\subset \Lambda_0$ can be parameterized as
\[
u^*\equiv u^*(x,y)=(x,y,z^*,p_x,p_y,p_z^*)\,,
\]
where $p_x$ and $p_y$ are fixed and $(x,y)\in \TT^2$.
Moreover, the stable manifold $W^\srm(\cT_{p_x,p_y})$ 
is given by the set of points of the form
\begin{equation}\label{eq:param:u0}
u^0\equiv u^0(\tau,x,y)=(x+F_1(\tau),y+F_2(\tau),z^0(\tau),p_x,p_y,p_z^0(\tau))\,,
\end{equation}
where $\tau\in\RR$, $(x,y)\in\TT^2$, the functions $z^0$ and $p_z^0$ are defined in Eq.~\eqref{eqsep}, and
\[
F_1(\tau) := \sin(z^*) \tau - \int_0^\tau \sin(z^0(\sigma))
d\sigma\,, \qquad
F_2(\tau) := \cos(z^*) \tau - \int_0^\tau \cos(z^0(\sigma))
d\sigma\,.
\]
Finally, we introduce some notation that will be useful in
Section~\ref{ssec:poincare}. If $\phi_t^0$ is the flow of the
Hamiltonian system $H_0$ and we consider points
$u^*\in \Lambda_0$ and $u^0 \in W^\srm (\Lambda_0)=W^\urm (\Lambda_0)$, 
then
\begin{align}
\phi_t^0 (u^*) = {} & (x+\omega_1 t,y+\omega_2 t,z^*,p_x,p_y,p_z^*), \label{eq:flow:H0:inner}\\
\phi_t^0 (u^0) = {} & (x+F_1(\tau+t)+\omega_1 t,y+F_2(\tau+t)+\omega_2 t,z^0(\tau+t),p_x,p_y,p_z^0(\tau+t)), \label{eq:flow:H0:outer}
\end{align}

We observe that the functions $F_1$ and $F_2$ depend on the constants
$p_x$ and $p_y$ through $z^*$ and $z^0$, but we omit this dependence in order to avoid cumbersome notation. After straightforward
computations we obtain the following explicit formulas
\begin{equation}\label{eq:F1:F2:ABC}
\begin{split}
F_1(\tau) & = \bigg( \frac{2(\tanh (\lambda \tau) -1)}{\lambda} +
\frac{2}{\lambda} \bigg) \sin z^* - \bigg(\frac{2 \mathrm{sech} (\lambda
\tau)}{\lambda}-\frac{2}{\lambda} \bigg) \cos z^*\,, \\
F_2(\tau) & = \bigg( \frac{2(\tanh (\lambda \tau) -1)}{\lambda} +
\frac{2}{\lambda} \bigg) \cos z^* + \bigg(\frac{2 \mathrm{sech} (\lambda
\tau)}{\lambda}-\frac{2}{\lambda} \bigg) \sin z^*\,,
\end{split}
\end{equation}
where the constant $\lambda$ is defined in Eq.~\eqref{eq:lambda}. These functions allow us to compute the \emph{phase shift} of any trajectory when traveling along
$W^\srm(\cT_{p_x,p_y})$.
Indeed, the phase-shift is defined by the limits
\[
\begin{array}{ll}
x_+:= \lim_{t\rightarrow \infty} F_1(\tau+t)\,, \quad &  
x_-:= \lim_{t\rightarrow -\infty} F_1(\tau+t)\,,\\
y_+:= \lim_{t\rightarrow \infty} F_2(\tau+t)\,, \quad& 
y_-:= \lim_{t\rightarrow -\infty} F_2(\tau+t)\,,\\
\end{array}
\]
which can be explicitly computed and do not depend on $\tau$, that is
\begin{equation}\label{eq:phase-shift}
x_\pm = 2 \frac{(\mp p_x-p_y)}{(p_x^2 + p_y^2)^{3/4}}\,, \qquad 
y_\pm = 2 \frac{(p_x\mp p_y)}{(p_x^2 + p_y^2)^{3/4}}\,.
\end{equation}

Observe that the limits $x_+$ and $x_-$ are different, which means that
any point in the homoclinic orbit approaches different points 
of the same invariant torus if we consider the limit in the future and in the past.
This is the reason why the terminology phase-shift is used for this phenomenon, see e.g.~\cite{BT99,DLS00,FGKR13}.
As we will show in Section~\ref{sec:chains},
this phase-shift contributes to the
expression involved in the transversality conditions used to
obtain diffusion.

\begin{remark}\label{R:magsurf}
It is interesting to note that the invariant tori $\cT_{p_x,p_y}$
project onto the toroidal magnetic surfaces $z=z^*$ of the unperturbed ABC magnetic field in the
configuration space $\TT^3$. Moreover, the magnetic field on each
surface is linear, i.e. ${\bf B}_{ABC}|_{z=z^*}=(\sin z^*,\cos
z^*,0)$, and the trajectories follow the magnetic lines.
Let us observe that
the slope of the magnetic lines
$\tan z^*$ coincides with the ratio of the frequencies
$\omega_1/\omega_2$, cf. Eqs.~\eqref{eq:omega1} and~\eqref{eq:omega2}.
\end{remark}

\subsection{Main Theorem: diffusion along a NHIM}\label{ssec:mainresult}

Let us consider the following Hamiltonian for the ABC system
\begin{equation}\label{eq:scaled:ham}
H  = \frac{1}{2}(p_x-\sin z - \ep \hat C \cos y)^2 + \frac{1}{2}(p_y
-\cos z - \ep \hat B \sin x)^2 + \frac{1}{2}(p_z - \ep \hat C \sin y - \ep \hat B \cos x)^2,
\end{equation}
where we have introduced a scaling $B=\ep \hat B$ and $C = \ep \hat C$.
The following result states sufficient conditions for the
existence of diffusing orbits:

\begin{theorem}\label{teo:diffusion:ABC}
Consider the Hamiltonian~\eqref{eq:scaled:ham} of the ABC system with $\hat B\geq \hat C \neq 0$.
Assume that the following hypotheses hold:
\begin{itemize}
\item [$\mathbf{A}_1$] Considering the notation introduced in Section~\ref{ssec:geom:unper},
we define the functions
$M_i^0\equiv M_i^0(p_x,p_y)$ as
\begin{align*}
M_1^0 := {} & \hat B \int_{-\infty}^\infty \bigg( 
(p_y - \cos z^*) \sin (x_{\pm}+\omega_1 \sigma) 
-(p_y-\cos z^0) \sin (F_1 + \omega_1 \sigma) 
- p_z^0 \cos(F_1 + \omega_1 \sigma) 
\bigg) d\sigma, \\
M_2^0 := {} & \hat C \int_{-\infty}^\infty \bigg( 
(p_x - \sin z^*) \cos (y_{\pm}+\omega_2 \sigma) 
-(p_x-\sin z^0) \cos (F_2 + \omega_2 \sigma) 
- p_z^0 \sin(F_2 + \omega_2 \sigma) 
\bigg) d\sigma, \\
M_3^0 := {} & \hat B \int_{-\infty}^\infty \bigg( 
(p_y - \cos z^*) \cos (x_{\pm}+\omega_1 \sigma) 
-(p_y-\cos z^0) \cos (F_1 + \omega_1 \sigma) 
+ p_z^0 \sin(F_1 + \omega_1 \sigma) 
\bigg) d\sigma, \\
M_4^0 := {} &  \hat C \int_{-\infty}^\infty \bigg( 
(p_x-\sin z^0) \sin (F_2 + \omega_2 \sigma) 
- (p_x - \sin z^*) \sin (y_{\pm}+\omega_2 \sigma) 
- p_z^0 \cos(F_2 + \omega_2 \sigma) 
\bigg) d\sigma, 
\end{align*}
with $F_1\equiv F_1(\sigma)$ and $F_2\equiv F_2(\sigma)$,
and where the notation $x_{\pm}$ (resp. $y_{\pm}$) means that we take $x_-$ (resp. $y_-$) when we integrate in the interval $(-\infty,0)$, and $x_+$ (resp. $y_+$) when we integrate in the interval $(0,\infty)$. We assume that there exists a non-empty set
$\cI=[a_1,b_1] \times [a_2,b_2]$, 
for positive values of $a_i,b_i$, such that $M_1^0$ and $M_3^0$ do not vanish simultaneously, and the same for $M_2^0$ and $M_4^0$, provided that $(p_x,p_y)\in\cI$.

\item [$\mathbf{A}_2$]
Assume that for any value $(p_x,p_y) \in \cI$ there exists
a non-empty domain $\cJ_{p_x,p_y} \subset \TT^2$ with the property that
\[
\cD:= \bigcup_{(p_x,p_y) \in \cI} \cJ_{p_x,p_y} \times \{(p_x,p_y)\} \subset \TT^2 \times \cI
\]
is a domain, and when $(x,y,p_x,p_y)\in \cD$ there is a unique critical point
$\tau^*\equiv\tau^*(x,y,p_x,p_y)$ of the map
\[
\tau \mapsto M_1^0 \cos(x-\omega_1 \tau) + M_2^0 \cos(y-\omega_2 \tau) + M_3^0 \sin(x-\omega_1 \tau) + M_4^0 \sin (y-\omega_2\tau)\,,
\]
which defines a smooth function on $\cD$.

\item [$\mathbf{A}_3$] 
Assume that we can chose a constant $L>0$ such that for every $(x,y,p_x,p_y)\in \cD$
we have
\begin{equation}\label{eq:cond:trans:teo}
\left\{
\begin{array}{ll}
\Delta_1 \Delta_3 - \Delta_2^2 \neq 0, & \mbox{if $|p_x-p_y|\geq L$}\,, \\
\hat \Delta_1 \hat \Delta_4 -\hat \Delta_2 \hat \Delta_3 \neq 0, & \mbox{if $|p_x-p_y|\leq L$}\,,
\end{array}
\right.
\end{equation}
where $\{\Delta_i\}_{i=1,2,3}$ and $\{\hat \Delta_i\}_{i=1,2,3,4}$ are certain explicit functions depending on $(x,y,p_x,p_y)$ that are defined in Section~\ref{sec:chains}, cf. Eqs.~\eqref{eq:Delta:1}--\eqref{eq:Delta:3} and~\eqref{eq:Delta:4}--\eqref{eq:Delta:7}.
\end{itemize}

Then, given two pairs
$(p_x^{0},p_y^{0}) \in \mathring{\cI}$ and
$(p_x^{1},p_y^{1}) \in \mathring{\cI}$
and given $\delta>0$,
there exists $\ep^*=\ep^*(\delta,\cI)$
such that if $0<|\ep|< \ep^*$ then
there is a trajectory $(x(t),y(t),z(t),p_x(t),p_y(t),p_z(t))$
of the system~\eqref{eq:scaled:ham} satisfying
\begin{align*}
& \mathrm{dist}\left((p_x^{0},p_y^{0}),(p_x(0),p_y(0))\right) \leq \delta, \\
& \mathrm{dist}\left((p_x^{1},p_y^{1}),(p_x(T),p_y(T))\right) \leq \delta.
\end{align*}
for some $T>0$.
\end{theorem}

We would like to emphasize that the above hypotheses are given in a very explicit way. To evaluate all the functions involved in the statement of Theorem~\ref{teo:diffusion:ABC}, we only need to compute
the coefficients $\{M_i^0\}_{i=1,2,3,4}$ in Hypothesis $\mathbf{A}_1$, together with the 
partial derivatives
$\tau^*_{x}$,
$\tau^*_{y}$,
$\tau^*_{xx}$,
$\tau^*_{xy}$
and $\tau^*_{yy}$, of the critical point in Hypothesis $\mathbf{A}_2$.
As was sketched in the introduction,
the proof of Theorem~\ref{teo:diffusion:ABC} consists
in combining the internal dynamics on the NHIM with its outer (asymptotic)
dynamics. Details are presented and discussed in Sections~\ref{sec:inner},
\ref{sec:outer}, and \ref{sec:chains}.
Finally, in Section~\ref{sec:explicit} we show that the hypotheses of the theorem can be rigorously checked in a computer assisted proof.

\begin{remark}
Since the invariant tori $\cT_{p_x,p_y}$ correspond to the toroidal magnetic surfaces of the unperturbed ABC magnetic field, c.f. Remark~\ref{R:magsurf}, Theorem~\ref{teo:diffusion:ABC} implies the existence of drift motions connecting any two magnetic surfaces (compatible with the set $\cI$) for the perturbed ABC system. This diffusion of charged particles is a very harmful phenomenon for the confinement of hot plasmas for fusion power generation, as explained in the introduction.  
\end{remark}

\section{Inner dynamics of the normally hyperbolic invariant manifold}
\label{sec:inner}

The study of normally hyperbolic invariant manifolds is a very
classical (and important) topic and it has been extensively
considered in the literature. Most of the results that we use
in this section are standard and can be found in~\cite{Fenichel,HPS}. Our purpose
here is to present a basic overview, notation and
perturbative formulas that we require to study the perturbation
of the normally hyperbolic invariant manifold introduced in
Section~\ref{ssec:geom:unper}.

We recall that our goal is to study the Hamiltonian~\eqref{eq:scaled:ham}
for small values of $\ep$. Hence, we
write $H=H_\ep$
perturbatively as follows
\begin{equation}\label{eq:Ham:exp}
H_\ep= H_0 + \ep H_1 + \ep^2 H_2,
\end{equation}
where
\begin{align}
H_0 = {} & \frac{1}{2} (p_x - \sin z)^2 + \frac{1}{2}(p_y - \cos z)^2 + \frac{1}{2} p_z^2,
\label{eq:H0:two} \\
H_1 = {} & - \hat C \cos y(p_x - \sin z) - \hat B \sin x (p_y - \cos z) - p_z(\hat C \sin y
+ \hat B \cos x),
\label{eq:H1} \\
H_2 = {} & \frac{\hat C^2}{2}+\frac{\hat B^2}{2}+\hat B \hat C \cos x \sin y.
\label{eq:H2}
\end{align}

The unperturbed Hamiltonian $H_0$ was studied in Section~\ref{ssec:geom:unper},
where we characterized the corresponding NHIM $\Lambda_0$. 
Now, we are interested in characterizing the perturbed invariant manifold $\Lambda_\ep$ together
with the restricted dynamics on it (mainly the existence
and approximation of invariant tori). To this end, we will follow closely the methodology introduced in the
papers~\cite{DLS06,DLS08,DLS13}.

Let us remark that the Hamiltonian~\eqref{eq:Ham:exp} is real-analytic. This will
imply that all the objects obtained in this section will be of class $C^r$, with
arbitrarily large $r$ (this follows from Fenichel rate conditions) so that
we can omit all the discussions concerning regularity. This will simplify
many technical issues, for example when applying averaging and KAM theory.
The interested reader is referred to~\cite{DLS06,DLS13} for details on regularity.

\subsection{Normally hyperbolic invariant manifolds and perturbative setting}\label{ssec:inner:approx}

Let $M$ be a smooth finite dimensional manifold and let us
consider a flow $\phi_t$, of class $C^r$ with $r\geq1$, acting 
on $M$.

\begin{definition}
Let $\Lambda\subset M$ be a submanifold invariant under
the flow, i.e., $\phi_t (\Lambda)=\Lambda$. We say that $\Lambda$
is a normally hyperbolic invariant manifold (NHIM), if there
exist a constant $c>0$, expansion rates
$0<\mu<\lambda$, and 
a splitting for every $x\in \Lambda$
\begin{equation}\label{eq:splitting}
T_x M = E^\srm_x \oplus E^\urm_x \oplus T_x\Lambda,
\end{equation}
characterized as follows
\begin{equation}\label{eq:NHIM:prop}
\begin{split}
v \in E^\srm_x  &\Longleftrightarrow |\Dif \phi_t(x)v| \leq c\, \rme^{-\lambda t\hphantom{||}}|v|, \qquad t\geq 0, \\
v \in E^\urm_x  &\Longleftrightarrow |\Dif \phi_t(x)v| \leq c\, \rme^{-\lambda |t|}|v|,  \qquad t\leq 0, \\
v \in T_x\Lambda&  \Longleftrightarrow |\Dif \phi_t(x)v| \leq c\, \rme^{\mu|t|}|v|,  \qquad t\in \RR.
\end{split}
\end{equation}
\end{definition}

The classical theory of NHIMs guarantees that if $\Lambda$ is normally hyperbolic, then it
is persistent under small perturbations. Moreover, if the
system depends smoothly on parameters, the manifolds |they may not be
unique| can be chosen to depend smoothly on parameters.
NHIMs are robust under perturbations, so we do not
require a symplectic structure on $M$ and $\phi_t$. Nevertheless,
the problem considered in this paper is endowed with a symplectic structure
and hence we will be interested in characterizing a symplectic structure on the perturbed NHIM.

In order to apply the geometric mechanism for a priori unstable systems
(c.f.~\cite{DLS06,DLS13}) we must compute explicitly some expansions in $\ep$ of
the NHIM associated to the Hamiltonian~\eqref{eq:Ham:exp}.
Notice
that in our case we can model the NHIM by means of the canonical manifold
$N=\TT^2 \times \RR^2$ (see Section~\ref{ssec:geom:unper}), that is,
we look for a parameterization $P_\ep: N \rightarrow M$,
with $P_\ep(N)=\Lambda_\ep$, characterized by the invariance
equation
\begin{equation}\label{eq:NHIM:inv}
X_{\ep} \circ P_\ep = \Dif P_\ep R_\ep
\end{equation}
where $R_\ep$ is a vector field on $N$ and $X_\ep$ is the
Hamiltonian vector field associated to $H_\ep$.
Using the expansions
\begin{align*}
X_\ep = {} & X_0 + \ep X_1 + \ep^2 X_2 + \ldots, \\
P_\ep = {} & P_0 + \ep P_1 + \ep^2 P_2 + \ldots, \\
R_\ep = {} & R_0 + \ep R_1 + \ep^2 R_2 + \ldots,
\end{align*}
and equating terms in the expansion of $\ep$ of the invariance equation~\eqref{eq:NHIM:inv},
we obtain (this approach was used in~\cite{DLS06})
\begin{align}
\mbox{0th order:} & \quad X_{0} \circ P_0 = \Dif P_0 R_0,
\label{eq:param0}\\
\mbox{1st order:} & \quad (\Dif X_{0} \circ P_0) P_1 + X_{1} \circ P_0 =
\Dif P_0 R_1 +\Dif P_1 R_0, \label{eq:param1}\\
\mbox{2nd order:} & \quad (\Dif X_{0} \circ P_0) P_2 + \frac{1}{2}
(\Dif^2X_{0} \circ P_0) P_1^{2\otimes} + (\Dif X_{1} \circ P_0) P_1 + X_{2} \circ
P_0  \label{eq:param2} \\
& \quad \quad = \Dif P_0 R_2 + \Dif P_1 R_1 + \Dif P_2 R_0, \nonumber \\
\mbox{$n$th order:} & \quad (\Dif X_{0}\circ P_0) P_n - \Dif P_n R_0 - \Dif P_0
R_n = -X_{n} \circ P_0 + S_n,\label{eq:paramn}
\end{align}
where $S_n$ is a polynomial in $X_{0}, \ldots, X_{n-1}$, their
derivatives, $P_0,\ldots,P_{n-1}$, their derivatives, and $R_0,\ldots,R_{n-1}$.

Clearly (see the discussion in Section~\ref{ssec:geom:unper})
Eq.~\eqref{eq:param0} has the solution
\begin{align*}
P_0(x,y,p_x,p_y) & =(x,y,z^*,p_x,p_y,p_z^*)\\
R_0(x,y,p_x,p_y) & = \omega_1(p_x,p_y) \partial_x + \omega_2(p_x,p_y) \partial_y
\end{align*}
where $\omega_1$ and $\omega_2$ are given by~\eqref{eq:omega1}
and~\eqref{eq:omega2}, respectively. 
In this case, since the unperturbed internal field $R_0$ does
not depend on the angular variables $(x,y)$, the equations of the form~\eqref{eq:paramn}
lead to simple cohomological equations in a suitable frame. Hence,
these equations can be
solved explicitly using Fourier expansions.
It is worth mentioning that there are more general theories that allow us to solve
equations of the form~\eqref{eq:paramn} even if the motion on the base
is not quasi-periodic.

As will be discussed in subsequent sections, the solution of
equations~\eqref{eq:param0},~\eqref{eq:param1},~\eqref{eq:param2}, and~\eqref{eq:paramn}
is not uniquely determined.
We will use this freedom in order to
obtain certain symplectic properties. More specifically, we follow
the ideas in~\cite{DLS08} to maintain the canonical symplectic
structure on $\Lambda_\ep$, so that we can easily characterize
and manipulate the Hamiltonian associated to the restricted vector field $R_\ep$.

\subsection{Symplectic properties of NHIMs of Hamiltonian systems}\label{ssec:NHIM:symp}

Let $M$ be a symplectic manifold with symplectic form $\sform$, represented by
a matrix-valued function $\Omega$, and let us assume that a $C^r$ Hamiltonian $H_0$, with $r\geq 2$, has a NHIM $\Lambda_0$ parameterized by
$P_0:N \rightarrow M$. Then, it is well known
(c.f.~\cite{Fenichel,HPS}) that for every perturbed Hamiltonian
$H_\ep$ of class $C^r$ there exists a NHIM $\Lambda_\ep$ parameterized
by $P_\ep$ of class $C^{r-1}$. Moreover, $\Lambda_\ep$
is $\cO(\ep)$-close to $\Lambda_0$ in the $C^{r-2}$ sense. Here and in what follows, when we say that a map depending on parameters is of class $C^r$ we shall mean that it is of class $C^r$ in all variables including the parameters.

Given a family of Hamiltonians 
having a family of NHIMs $\Lambda_\ep = P_\ep(N)$, with $P_\ep :N \rightarrow M$, we consider the maps $R_\ep : N \rightarrow TN$
corresponding to the vector fields restricted to the NHIMs. The maps $P_\ep$ and $R_\ep$ are
related by the invariance equation~\eqref{eq:NHIM:inv}.

It is well known that the solutions of~\eqref{eq:NHIM:inv} are not uniquely
defined, since we have the possibility of choosing different coordinates in the
reference manifold $N$. It is natural to use this freedom to satisfy certain properties,
like asking $P_\ep$ to be a graph or asking $R_\ep$ to be as simple
as possible.
In this paper, we are interested in choosing the solution that preserves
the Hamiltonian structure of the problem, that is, we want that
\begin{equation}\label{eq:Pepsymp}
\frac{d}{d\ep}(P_\ep^* \sform)=0.
\end{equation}
The fact that this can be achieved 
was proved in~\cite{DLS08}.
In this paper, since we need to perform some explicit
computations, we have to give some additional details on the
procedure presented in~\cite{DLS08}. The aim of this section is to explain 
the explicit computations required to handle a particular problem.

A natural way to obtain~\eqref{eq:Pepsymp} is to use deformation
theory. Let us recall some standard definitions.
Given two connected manifolds $M$ and $N$, 
and given a family $f_\ep : N \rightarrow M$ such that $(x,\ep) \mapsto f_\ep(x)$
is $C^1$ in all its arguments, we define the infinitesimal deformation
of $f_\ep$ as the vector field $\cF_\ep$ that satisfies
\[
\frac{d}{d\ep}f_\ep=\cF_\ep \circ f_\ep,
\]
and we observe that $\cF_\ep=(\frac{d}{d\ep} f_\ep)\circ f_\ep^{-1}$ is
defined on $f_\ep(N) \subset M$.

Let $\cP_\ep$ be the infinitesimal deformation of the family $P_\ep$
with initial condition $P_0$. It is clear that $\cP_\ep : \Lambda_\ep \rightarrow TM$,
so we can consider the projections of $\cP_\ep$ according to the splitting~\eqref{eq:splitting}. Then
we have the following result \cite{DLS08}:
\begin{proposition}\label{prop:NHIM:def}
Let us consider a family of parameterizations $P_\ep : N \rightarrow M$ with
$\Lambda_\ep = P_\ep(N)$. Assume that the infinitesimal deformation $\cP_\ep$
satisfies that the projection on the space $T_x \Lambda_\ep$ vanishes for
every $x\in \Lambda_\ep$. Then, the symplectic form $P_\ep^* \sform_{*,\ep}$ is independent
of $\ep$, where $\sform_{*,\ep}$ is the original form $\sform$ expressed
in a basis of the splitting~\eqref{eq:splitting}.
\end{proposition}

\begin{proof}
For the sake of completeness, we reproduce the proof given in~\cite{DLS08}.
First we observe that since $\sform$ is invariant under the flow $\phi^\ep_t$ of $H_\ep$,
then also is $\sform_{*,\ep}$, and we have
\[
\sform_{*,\ep}(x)[u,v] = \sform_{*,\ep}(\phi_t^\ep(x))[\Dif \phi_t^\ep(x)u,\Dif \phi_t^\ep(x)v],
\]
for every $u,v \in T_xM$ and $t\in \RR$. Using the asymptotic properties in~\eqref{eq:NHIM:prop}
it is clear that $\sform_{*,\ep}(x)[u,v]=0$ if $u \in E^\srm_{x,\ep} \oplus E^\urm_{x,\ep}$ and
$v \in T_x \Lambda_{\ep}$ (or vice versa).

Then, using Cartan's formula we obtain
\[
\frac{d}{d\ep} P_\ep^* \sform_{*,\ep} = P_\ep^* (i_{\cP_\ep} d\sform_{*,\ep} + d i_{\cP_\ep} \sform_{*,\ep}) = P_\ep^* d i_{\cP_\ep} \sform_{*,\ep},
\]
where we used that $\sform_{*,\ep}$ is closed. Then, we have
\[
\frac{d}{d\ep}(P_\ep^* \sform_{*,\ep})=dP_\ep^*i_{\cP_\ep} \sform_{*,\ep}
\]
and we observe that the 1-form $P_\ep ^* i_{\cP_\ep} \sform_{*,\ep}$, acting on $v \in T_xN$,
is given by
\[
P_\ep^* i_{\cP_\ep} \sform_{*,\ep} (x) [v] = i_{\cP_\ep} \sform_{*,\ep} (P_\ep(x))[dP_\ep(x) v]= \sform_{*,\ep}(P_\ep(x))[\cP_\ep(P_\ep(x)),dP_\ep(x)v].
\]
By hypothesis, we have $\cP_\ep(P_\ep(x)) \in E^\srm_{P_\ep(x),\ep} \oplus E^\urm_{P_\ep(x),\ep}$
and we also have $dP_\ep(x) v \in T_{P_\ep(x)}\Lambda_\ep$. Hence, it must be $P_\ep^* i_{\cP_\ep} \sform_{*,\ep} (x) \equiv 0$ and we conclude that $P_\ep^* \sform_{*,\ep}$ is independent of $\ep$.
\end{proof}

\begin{remark}
A particularly interesting case arises if
$\Lambda_0=P_0(N)$ is a NHIM for $X_0$ and Eq.~\eqref{eq:NHIM:inv} is solved
perturbatively. This is the situation considered in this paper.
Property~\eqref{eq:Pepsymp} is important in order to have a canonical
symplectic structure on $\Lambda_\ep$, so that the averaging procedure
(normal form) can be implement in the usual way.
\end{remark}
In the following we assume that $M=N \times \TT \times \RR$, with $N=\TT^{n} \times \RR^{n}$, and we
use the notation $(u,p_u) \in N$ with $u=(u_1,\ldots,u_n)$, $p_u=(p_{u,1},\ldots,p_{u,n})$, and $(v,p_v)\in \TT \times \RR$.
We endow $M$ with the symplectic form
\begin{equation}\label{eq:sform:canonic}
\sform = \sum_{i=1}^n dp_{u,i} \wedge d u_i + dp_v \wedge dv,
\end{equation}
which is represented by
\[
\Omega_{n+1}=
\begin{pmatrix}
\Omega^0_{n} & O_{2n \times 2} \\
O_{2 \times 2 n} & \Omega^{0}_{1}
\end{pmatrix},
\quad
\mbox{with}
\quad
\Omega^0_{n} =
\begin{pmatrix}
O_n & -I_n \\
I_n & O_n
\end{pmatrix}
\]
where from now on we use the notation 
$O_{n \times m}$, $I_{n\times m}$,
$O_n \equiv O_{n \times n}$, and $I_n\equiv I_{n\times n}$,
for the zero and identity matrices, respectively.
Moreover, we denote by $\mathrm{M}_{m \times n}$
the space of $m \times n$-matrices with real coefficients.

\begin{definition}\label{def:symp:frame}
Given a parameterization $P_0:N \rightarrow M$
of a NHIM, with $N=\TT^n \times \RR^n$ and $M=N \times \TT\times \RR$, we say that $P_0$ is compatible with the symplectic form $\sform$ if
\[
\Dif P_0 (u,p_u)^\top \Omega_{n+1} \Dif P_0(u,p_u)= \Omega^0_{n}.
\]
Similarly, we say that a frame
\[
\begin{array}{rcl}
\mathfrak{C}: 
N \times \RR^{2n +2} & 
\longrightarrow &
T_{P_0(N)}M \\
(u,p_u,\xi) &\longrightarrow & (P_0(u,p_u),C_0(u,p_u) \xi)
\end{array},
\]
with $C_0: N\rightarrow \mathrm{M}_{(2n+2) \times (2n+2)}$,
is symplectic if
\[
C_0(u,p_u)^\top \Omega_{n+1} C_0(u,p_u)= \Omega_{n+1}\,.
\]
\end{definition}

Let us also introduce some notation regarding derivatives
that will be useful in computations. Given a vector field $R$
on a NHIM, and given a function $\xi : N \rightarrow \RR$,
we denote the Lie derivative of $\xi$ with respect to $R$ as follows
\begin{equation}\label{eq:Lie}
\Lie{R}(\xi) = \Dif \xi R = \sum_{i=1}^{n} \frac{\partial \xi}{\partial_{u_i}} R_i + \sum_{i=1}^{n} \frac{\partial \xi}{\partial_{p_{u,i}}} R_{n+i}\,.
\end{equation}
Moreover, given a parameterization $P:N \rightarrow M$, and vector fields $X$ and $R$ on $M$ and $N$, respectively, we introduce the operator
\begin{equation}\label{eq:cR:general}
\cR_{P,X,R}(\xi)=\Dif X \circ P \xi - \Lie{R}(\xi)\,,
\end{equation}
acting on functions $\xi : N \rightarrow \RR$.
We extend the notation in~\eqref{eq:Lie} and~\eqref{eq:cR:general} component-wise for matrix functions $\xi : N \rightarrow \mathrm{M}_{m\times n}$. In other to simplify the notation, we will write $\cR_0\equiv \cR_{P_0,X_0,R_0}$.

Given a parameterization $P_0:N \rightarrow M$ of a NHIM, with $N=\TT^n \times \RR^n$ and $M=N \times \TT\times \RR$,
we can take derivatives at both sides of the invariance equation $X_0 \circ P_0 = \Dif P_0 R_0$
thus obtaining
\[
\Dif X_0 \circ P_0 \Dif P_0 = \Dif(\Dif P_0 R_0)=\Lie{R_0} (\Dif P_0) + \Dif P_0 \Dif R_0.
\]
This means that the tangent vectors of $P_0(N)$ partially characterize the
action of the operator $\cR_0$ in~\eqref{eq:cR:general}
as
\[
\cR_0(\Dif P_0)=\Dif P_0 \Dif R_0.
\]
Since $P_0(N)$ is normally hyperbolic, there exist maps $W_0 : N \rightarrow \mathrm{M}_{(2n+2)\times 2}$ parameterizing the normal bundle of $P_0(N)$, and
$\Gamma_0: N \rightarrow \mathrm{M}_{2\times 2}$ such that
\[
\cR_0(W_0)=W_0 \Gamma_0.
\]
From now on, we assume that $\Gamma_0$ is diagonal, and due to the
Hamiltonian structure we can write 
$$\Gamma_0=\begin{pmatrix}
\lambda_0&0\\0&-\lambda_0
\end{pmatrix}
\,.$$
Moreover, if we assume
that $P_0$ is compatible with the symplectic form $\sform$, then
it turns out that the matrix $W_0$ can be scaled in such a way
that the juxtaposed matrix $C_0:=(\Dif P_0~W_0)\in \rm M_{(2n+2)\times(2n+2)}$ defines a symplectic
frame as in Definition~\ref{def:symp:frame}.

The operator $\cR_0$ introduced above appears in
the perturbative equations~\eqref{eq:param0}--~\eqref{eq:paramn}
obtained in Section~\ref{ssec:inner:approx}. The following lemma
approaches the study of 
these equations
using the previously constructed frame. It is worth mentioning
that the fact the frame $\mathfrak{C}$ is assumed to be
symplectic is not really necessary. Nevertheless, it simplifies
some computations (for example the computation of the inverse $C_0^{-1}$).

\begin{lemma}\label{lem:coho}
Assume that $P_0:N \rightarrow M$ satisfies $X_0 \circ P_0 = \Dif P_0 R_0$,
with $N=\TT^n \times \RR^n$ and $M=N \times \TT\times \RR$. 
Given a map $\eta : N \rightarrow \RR^{2n+2}$, we consider the
following equation
\begin{equation}\label{eq:paramg}
\Dif X_0 \circ P_0 \xi - \Dif \xi R_0 - \Dif P_0 \rho = \eta
\end{equation}
for the unknowns $\xi: N \rightarrow \RR^{2n+2}$ and $\rho : N \rightarrow \RR^{2n}$.
Then, using the symplectic frame $\mathfrak{C}$ associated to the
matrix $C_0=(\Dif P_0~W_0)$ constructed above, it turns out that Eq.~\eqref{eq:paramg}
leads to
\begin{align}
- \Lie{R_0} (\hat \xi^C) + \Dif R_0 \hat \xi^C = {} & \hat \eta^C + \rho \label{eq:coho1}\\ 
- \Lie{R_0} (\hat \xi^H) + \Gamma_0 \hat \xi^H = {} & \hat \eta^H \label{eq:coho2}
\end{align}
where 
\[
\xi=C_0 \hat \xi = \Dif P_0 \hat \xi^C + W_0 \hat \xi^H~~\mbox{and}~~\hat \eta=
\begin{pmatrix}
\hat \eta^C \\
\hat \eta^H
\end{pmatrix}
= - \Omega_{n+1} C_0^\top \Omega_{n+1}\eta\,,
\]
with $\hat \eta^C:N\rightarrow \RR^{2n}$, $\hat \xi^C:N\rightarrow \RR^{2n}$, $\hat \eta^H:N\rightarrow \RR^2$ and $\hat\xi^H:N\rightarrow \RR^2$.
\end{lemma}

\begin{proof}
Let us observe that the fact that $\mathfrak{C}$ is chosen to be symplectic
allows us to compute the inverse of $C_0$ as follows
\[
C_0^{-1}=\Omega_{n+1}^{-1} C_0^\top \Omega_{n+1}=-\Omega_{n+1} C_0^\top \Omega_{n+1}\,.
\]
We also notice that the action of $\cR_0$ on the matrix $C_0 \hat \xi$
takes the form
\[
\cR_0(C_0 \hat \xi) = \cR_0(C_0) \hat \xi - C_0 \Lie{R_0}(\hat \xi)\,,
\]
and that
\begin{align*}
C_0^{-1} \cR_0(C_0) = {} & - \Omega_{n+1} C_0^\top \Omega_{n+1} \left( \Dif X_0 \circ P_0 C_0 - \Lie{R_0}(C_0) \right) \\
= {} & - \Omega_{n+1}
\begin{pmatrix}
\Dif P_0^\top \Omega_{n+1} \Dif P_0 \Dif R_0 & \Dif P_0^\top \Omega_{n+1} W_0 \Gamma_0 \\
W_0^\top \Omega_{n+1} \Dif P_0 \Dif R_0 & W_0^\top \Omega_{n+1} W_0 \Gamma_0
\end{pmatrix} \\
= {} & - \Omega_{n+1} \Omega_{n+1}
\begin{pmatrix}
\Dif R_0 & O_{2n \times 2} \\
O_{2 \times 2n} & \Gamma_0
\end{pmatrix}
=\begin{pmatrix}
\Dif R_0 & O_{2n \times 2} \\
O_{2 \times 2n} & \Gamma_0
\end{pmatrix}\,.
\end{align*}
Introducing $\xi=C_0 \hat \xi=\Dif P_0 \hat \xi^C + W_0 \hat \xi^H$
into Eq.~\eqref{eq:paramg}, we obtain
\[
\begin{pmatrix}
\Dif R_0 & O_{2n \times 2} \\
O_{2 \times 2n} & \Gamma_0
\end{pmatrix}
\begin{pmatrix}
\hat \xi^C \\
\hat \xi^H 
\end{pmatrix}
-
\begin{pmatrix}
\Lie{R_0} (\hat \xi^C) \\
\Lie{R_0} (\hat \xi^H) 
\end{pmatrix}
-C_0^{-1}
\Dif P_0 \rho = C_0^{-1} \eta.
\]
Then, we observe that
\[
-C_0^{-1} \Dif P_0 \rho = \Omega_{n+1} C_0^\top \Omega_{n+1} \Dif P_0 \rho= \Omega_{n+1}
\begin{pmatrix}
\Dif P_0^\top \Omega_{n+1} \Dif P_0 \\
W_0^\top \Omega_{n+1} \Dif P_0
\end{pmatrix}\rho
=-
\begin{pmatrix}
I_{2n} \\
O_{2 \times 2n}
\end{pmatrix}\rho\,.
\]
Finally, using the symplectic structure, we introduce $\hat \eta^C$ and $\hat \eta^H$ as in the statement of the lemma, thus ending up with the equations~\eqref{eq:coho1} and~\eqref{eq:coho2}.
\end{proof}

It is standard to check that the solution of Eq.~\eqref{eq:coho2} is unique. In
our particular case (see computations in Section~\ref{ssec:NHIM:ABC}),
it turns out that $R_0$ produces an integrable quasi-periodic motion
in $N$, and hence, we can solve~\eqref{eq:coho2} using Fourier series.
In particular, if we have a function $\beta : N \rightarrow \RR$ expressed in
Fourier series as
\[
\beta(u,p_u)= \sum_{k\in \ZZ^n} \left( \beta^{\cos}_k(p_u) \cos(k \cdot u) + \beta^{\sin}_k(p_u) \sin(k\cdot u) \right),
\]
with $\beta^{\sin}_0 \equiv 0$, then it turns out that the solution $\xi$ of the equation
$\lambda \xi-\Lie{R_0}(\xi)=\beta$ is given by
\[
\xi(u,p_u)= \sum_{k\in \ZZ^n} \left( \xi^{\cos}_k(p_u) \cos(k \cdot u) + \xi^{\sin}_k(p_u) \sin(k\cdot u) \right),
\]
with
\begin{equation}\label{eq:solc:coho:general}
\xi^{\cos}_k = \frac{\beta^{\cos}_k \lambda+\omega \cdot \beta^{\sin}_k}{\lambda^2 + (\omega \cdot k)^2}\,,
\qquad
\xi^{\sin}_k = \frac{\beta^{\sin}_k \lambda-\omega \cdot \beta^{\cos}_k}{\lambda^2 + (\omega \cdot k)^2}\,.
\end{equation}
In case that $R_0$ takes a more general form, Eq.~\eqref{eq:coho2}
can be solved using the asymptotic properties of the cocycle.

As was mentioned in Section~\ref{ssec:inner:approx}, the solution of
Eq.~\eqref{eq:coho1} is not unique. A simple choice consists in taking
\begin{equation}\label{eq:particular}
\hat \xi^C = O_{2n \times 1}, \qquad \rho=-\hat \eta^C,
\end{equation}
but, in general, this solution will not determine a parameterization
which is compatible with the symplectic structure of the problem. The final goal
of this section is to compute the deformation of the symplectic frame $\mathfrak{C}$
with respect to the perturbation parameter and to combine Proposition~\ref{prop:NHIM:def}
and Lemma~\ref{lem:coho} in order to obtain the canonical symplectic structure in the deformed NHIM.

Assume that $P_\ep : N \rightarrow M$, with $N=\TT^n \times \RR^n$ and
$M=N \times \TT \times \RR$, is a family of parameterizations satisfying
$X_\ep \circ P_\ep = \Dif P_\ep R_\ep$, where $X_\ep$ is a family of Hamiltonian
vector fields with the symplectic form $\sform$ given by~\eqref{eq:sform:canonic}.
Let us consider $\cP_\ep$, the infinitesimal deformation of the family $P_\ep$
with initial condition $P_0$. A simple computation shows that
\[
\frac{d\cP_\ep}{d\ep}=P_0 + 2 P_2 \ep + 3 P_3 \ep^2 + \ldots = \cP_0 \circ P_0 + (\Dif \cP_0 \circ P_0 P_1 + \cP_1 \circ P_0) \ep +\ldots\,,
\]
thus obtaining
\begin{align}
\mbox{0th order:} & \qquad \hphantom{2n} P_1 = \cP_0 \circ P_0 \label{eq:defP1} \\
\mbox{1st order:} & \qquad \hphantom{n} 2 P_2 = \cP_1 \circ P_0 + \Dif \cP_0 \circ P_0 P_1 \label{eq:defP2} \\
\mbox{$n$th order:} & \qquad \hphantom{2} n P_n = \cP_n \circ P_0 + S_n \label{eq:defPn}
\end{align}
where $S_n$ is an explicit expression depending recursively on the previously
computed objects.

Let us consider the first order correction determined by Eq.~\eqref{eq:param1}.
We apply Lemma~\ref{lem:coho} with
\[
\xi=P_1, \qquad \rho=R_1, \qquad \eta=-X_1 \circ P_0
\]
and we consider the unique solution of Eqs.~\eqref{eq:coho1} and~\eqref{eq:coho2}
satisfying Eq.~\eqref{eq:particular}. In Eq.~\eqref{eq:defP1} we observe that $P_1$ is proportional to $\cP_0$. Hence, it turns out that
the deformation $\cP_0$ vanishes on the central directions. 
By Proposition~\ref{prop:NHIM:def}, we conclude that the reduced vector field $R_1$ 
is a Hamiltonian vector field with
respect to the form $\Omega_n^0$.

The second order correction is not so simple. On the one hand, we observe that
$P_2$ and $\cP_1$ are no longer proportional. On the other hand, we
have to consider Proposition~\ref{prop:NHIM:def} on the deformed
symplectic frame. Let us assume that we have computed $P_1$, $R_1$,
and also the first order correction of the symplectic frame, 
that is, $C_\ep=C_0+\ep C_1+\cO(\ep^2)$.
Then, we express the infinitesimal deformation $\cP_\ep$ on the frame $\mathfrak{C}_\ep$ perturbatively as
\begin{align*}
C_\ep^{-1} & \cP_\ep(P_\ep) = \\
& C_0^{-1} \cP_0 \circ P_0 + \ep(C_0^{-1} \Dif \cP_0 \circ P_0 P_1 
+ C_0^{-1} \cP_1 \circ P_0 - C_0^{-1} C_1 C_0^{-1} \cP_0 \circ P_0) + \cO(\ep^2).
\end{align*}
By construction, it is clear that
\[
C_0^{-1} \cP_0 \circ P_0 = C_0^{-1} P_1 = C_0^{-1} C_0 \hat \xi_1 = \hat \xi_1 =
\begin{pmatrix}
O_{2n\times 1} \\
\hat \xi_1^H
\end{pmatrix}.
\]
We ask the same condition for the $\ep$-order terms, thus obtaining that
\[
C_0^{-1} \Dif \cP_0 \circ P_0 P_1 + C_0^{-1} \cP_1 \circ P_0 - C_0^{-1} C_1 C_0^{-1} \cP_0 \circ P_0 =
\begin{pmatrix}
O_{2n\times 1} \\
\zeta
\end{pmatrix},
\]
for certain $\zeta : N \rightarrow \RR^2$ whose expression is not important for us.
Then, we use again that $\cP_0 \circ P_0 = P_1 = C_0 \hat \xi_1$, we
replace $\cP_1 \circ P_0$ using~\eqref{eq:defP2}, and we write $P_2=C_0 \hat \xi_2$,
thus obtaining the condition
\begin{equation}\label{eq:xi2:corre}
2 \hat \xi_2 - C_0^{-1} C_1  \begin{pmatrix}
O_{2n\times 1} \\
\hat \xi_1^H
\end{pmatrix} =
\begin{pmatrix}
O_{2n\times 1} \\
\zeta
\end{pmatrix}
\end{equation}
that determines the first $2n$ components $\hat \xi^C_2$ of $\hat \xi_2$.
Therefore, we can solve the second order correction of the invariance
equation, given by~\eqref{eq:param2}, using Lemma~\ref{lem:coho}
with 
\[
\xi=P_2, \qquad \rho=R_2, \qquad
\eta=-X_2 \circ P_0 +\Dif P_1 R_1 - \frac{1}{2} \Dif^2 X_0 \circ P_0 P_1^{2 \otimes} - DX_1 \circ P_1
\]
and choosing the unique solution obtained by fixing $\hat \xi^C_2$ satisfying Eq.~\eqref{eq:xi2:corre}.
Then, the corresponding correction of the reduced vector field, 
\begin{equation}\label{tonto}
R_2 = \Dif R_0 \hat \xi^C - \Lie{R_0}(\hat \xi^C) - \hat \eta^C,
\end{equation}
is a Hamiltonian vector field with
respect to the form $\Omega_n^0$.

Finally, we need to give a simple recipe to compute the first order correction $C_1$ of the
symplectic frame.
The construction
is analogous up to any order, but this is enough for our purposes. 
We will construct the frame taking $C_1=(\Dif P_1~W_1)$, where $W_1$ is computed as follows.
On the one hand,
we assume that we have computed $P_\ep=P_0 +\ep P_1 + \cO(\ep^2)$ so that we have
(the computation is direct)
\[
\cR_{P_\ep,X_\ep,R_\ep} (\Dif P_0+\ep \Dif P_1) = (\Dif P_0+\ep \Dif P_1) (\Dif R_0 + \ep \Dif R_1) + \cO(\ep^2),
\]
where we recall that $\cR_{P_\ep,X_\ep,R_\ep}$ is given by Eq.~\eqref{eq:cR:general}.
On the other hand, we look for $W_1$ and $\Gamma_1$ is such a way that
the action of $\cR_{P_\ep,X_\ep,R_\ep}$ on the matrix $W_0 + \ep W_1$ is given by
\[
\cR_{P_\ep,X_\ep,R_\ep} (W_0+\ep W_1) = (W_0+\ep W_1) (\Gamma_0 + \ep \Gamma_1) + \cO(\ep^2)\,.
\]
We observe that this condition is satisfied if
\begin{equation}\label{eq:cond:W1}
(\Dif X_0 \circ P_0) W_1 - \Lie{R_0} (W_1) - W_0 \Gamma_1 - W_1 \Gamma_0 = S_1\,,
\end{equation}
where
\[
S_1 := \Lie{R_1}(W_0)-\Dif X_1 \circ P_0 W_0 - \Dif^2 X_0 \circ P_0 P_1 \otimes W_0\,.
\]
Again, the solutions of Eq.~\eqref{eq:cond:W1} are obtained by considering the
action of the unperturbed operator $\cR_0$.
In the following result, analogous to Lemma~\ref{lem:coho}, we
study the above equation.

\begin{lemma}\label{lem:coho2}
Assume that $P_0:N \rightarrow M$ satisfies $X_0 \circ P_0 = \Dif P_0 R_0$,
with $N=\TT^n \times \RR^n$ and $M=N \times \TT\times \RR$. Assume that
the pair $P_1$ and $R_1$ is a solution of equation~\eqref{eq:param1}, that is,
we have
\[
(X_0+\ep X_1) \circ (P_0+\ep P_1)=(\Dif P_0+\ep \Dif P_1)(R_0+\ep R_1) + \cO(\ep^2).
\]
Then, using the symplectic frame $\mathfrak{C}$ associated to the
matrix $C_0=(\Dif P_0~W_0)$, it turns out that equation~\eqref{eq:cond:W1}
leads to
\begin{align}
-\Lie{R_0} (\hat W_1^C) + \Dif R_0 \hat W_1^C - \hat W_1^C \Gamma_0= {} & \hat S_1^C\,, \label{eq:coho3}\\ 
-\Lie{R_0} (\hat W_1^H) + \Gamma_0 \hat W_1^H-\hat W^H_1 \Gamma_0 = {} & \hat S_1^H - \Gamma_1\,,\label{eq:coho4}
\end{align}
where 
\begin{equation}\label{eq:not:split}
W_1=C_0 \hat W_1 = \Dif P_0 \hat W_1^C + W_0 \hat W_1^H~~\mbox{and}~~
\begin{pmatrix}
\hat S_1^C \\
\hat S_1^H
\end{pmatrix}
= - \Omega_{n+1} C_0^\top \Omega_{n+1} S_1\,.
\end{equation}
\end{lemma}

\begin{proof}
We recall that the frame $\mathfrak{C}$ satisfies
\[
\cR_0(C_0)=\Dif X_0\circ P_0 C_0 -\Dif (C_0) R_0 = C_0
\begin{pmatrix}
\Dif R_0 & O_{2n \times 2}\\
O_{2\times 2n} & \Gamma_0
\end{pmatrix}.
\]
Then, we compute the action of $\cR_0$ on $W_1=C_0 \hat W_1$ as follows
\[
\cR_0(C_0 \hat W_1) = C_0
\begin{pmatrix}
\Dif R_0 & O_{2n \times 2}\\
O_{2\times 2n} & \Gamma_0
\end{pmatrix} \hat W_1 - C_0 \Lie{R_0} (\hat W_1)\,,
\]
and we introduce this expression into~\eqref{eq:cond:W1}, thus obtaining
\[
C_0 \begin{pmatrix}
\Dif R_0 & O_{2n \times 2}\\
O_{2\times 2n} & \Gamma_0
\end{pmatrix} \hat W_1 - C_0 \Lie{R_0} (\hat W_1) - W_0 \Gamma_1 - C_0 \hat W_1 \Gamma_0 =S_1\,.
\]
Using the symplectic properties of the frame, we multiply both sides by
$C_0^{-1}=-\Omega_{n+1} C_0^\top \Omega_{n+1}$ and we end up with
\[
\begin{pmatrix}
\Dif R_0 & O_{2n \times 2}\\
O_{2\times 2n} & \Gamma_0
\end{pmatrix} \hat W_1 - \Dif (\hat W_1) R_0 + \Omega_{n+1} C_0^\top \Omega_{n+1} W_0 \Gamma_1 - \hat W_1 \Gamma_0 = -\Omega_{n+1} C_0^\top \Omega_{n+1} S_1\,.
\]
Finally, we observe that
\[
\Omega_{n+1} C_0^\top \Omega_{n+1} W_0 = \Omega_{n+1}
\begin{pmatrix}
\Dif P_0^\top \Omega_{n+1} W_0 \\
W_0^\top \Omega_{n+1} W_0
\end{pmatrix}
=
\begin{pmatrix}
O_{2n\times 2} \\
I_2
\end{pmatrix}\,,
\]
and using the notation in~\eqref{eq:not:split} we obtain the equations~\eqref{eq:coho3}
and~\eqref{eq:coho4}.
\end{proof}

Finally, we discuss the solution of equations~\eqref{eq:coho3}
and~\eqref{eq:coho4}. On the one hand, we observe that equation~\eqref{eq:coho3}
is similar to equation~\eqref{eq:coho2} in the sense that it can be solved
using Fourier series, obtaining a unique solution. On the other
hand, we observe that the diagonal part of the left hand side of equation~\eqref{eq:coho4}
is resonant. We can avoid this resonance by selecting $\Gamma_1$. To this end, we consider
the particular choice
\[
\Gamma_1 = \mathrm{diag} \langle \hat S_1^H \rangle\,,
\]
where $\langle \cdot \rangle$ stands for the average with respect to the variables $u \in \TT^n$.
Obviously, this choice preserves the diagonal character of
the matrix $\Gamma_\ep = \Gamma_0+ \ep \Gamma_1+\ldots$.

\subsection{Perturbative computation of the NHIM of the ABC system}
\label{ssec:NHIM:ABC}

The goal of this section is to compute the NHIM $\Lambda_\ep$ associated to
the ABC system in the perturbative setting given by Eqs.~\eqref{eq:Ham:exp}--~\eqref{eq:H2}. We follow the notation and methodology described
in Section~\ref{ssec:NHIM:symp}.

First, it is convenient to reorder the phase-space coordinates
as $(x,y,p_x,p_y,z,p_z)$ rather than $(x,y,z,p_x,p_y,p_z)$. In analogy with Section~\ref{ssec:NHIM:symp},
we have $(u,p_u)=(x,y,p_x,p_y) \in N=\TT^2 \times \RR^2$ and $(v,p_v)=(z,p_z)\in \TT \times \RR$. Then,
we consider the symplectic form $\sform$, and its matrix representation $\Omega_3$, given by
\[
\sform = dp_x \wedge dx + dp_y \wedge dy + dp_z \wedge d z, \qquad
\Omega_3 =
\begin{pmatrix}
\Omega^0_{2} & O_{4 \times 2} \\
O_{2 \times 4} & \Omega^{0}_{1}
\end{pmatrix}.
\]
With the above notation, we have $X_\ep=\Omega_3^{-1} \Dif H_\ep^\top=-\Omega_3 \Dif H_\ep^\top$.

We start with the explicit characterization of the unperturbed problem,
giving rise to the expressions
\[
X_0 =
\begin{pmatrix}
p_x - \sin z \\
p_y - \cos z \\
0 \\
0 \\
p_z \\
p_x \cos z - p_y \sin z
\end{pmatrix},
\qquad
P_0 =
\begin{pmatrix}
x \\
y \\
p_x \\
p_y \\
z^*=\arctan (p_x/p_y) + \pi \\
p_z^*=0
\end{pmatrix},
\]
and the corresponding derivatives
\[
\Dif X_0 \circ P_0 =
\begin{pmatrix}
0 & 0 &        1 &          0 &-\cos z^* &  0 \\
0 & 0 &        0 &          1 & \sin z^* &  0 \\
0 & 0 &        0 &          0 &0         &  0 \\
0 & 0 &        0 &          0 &0         &  0 \\
0 & 0 &        0 &          0 &0         &  1 \\
0 & 0 & \cos z^* & - \sin z^* &\lambda^2 &  0 
\end{pmatrix},~~
\Dif P_0 =
\begin{pmatrix}
1 & 0 & 0 & 0 \\
0 & 1 & 0 & 0 \\
0 & 0 & 1 & 0 \\
0 & 0 & 0 & 1 \\
0 & 0 & p_y \lambda^{-4} & -p_x \lambda^{-4} \\
0 & 0 & 0 & 0 
\end{pmatrix}\,,
\]
where we recall that 
$\lambda = (p_x^2 + p_y^2)^{1/4}$, 
$\sin z^* = - p_x/\lambda^2$, and 
$\cos z^* = - p_y/\lambda^2$
(see computations in Section~\ref{ssec:geom:unper})
Notice that the parameterization $P_0$ given above is compatible with the form $\sform$
according to Definition~\ref{def:symp:frame}.

To obtain the unperturbed symplectic frame we take the columns of $\Dif P_0$
and we complement them with the eigenvectors of $\Dif X_0 \circ P_0$
of eigenvalues $\lambda$ and $-\lambda$, that we suitable scale in
order to obtain a symplectic frame. Specifically, we take
\[
C_0 =
\begin{pmatrix}
1 & 0 & 0                & 0                 &  (\sqrt{2}/2)p_y \lambda^{-7/2} &  (\sqrt{2}/2)p_y \lambda^{-7/2} \\
0 & 1 & 0                & 0                 & -(\sqrt{2}/2)p_x \lambda^{-7/2} & -(\sqrt{2}/2)p_x \lambda^{-7/2} \\
0 & 0 & 1                & 0                 &                               0 & 0 \\
0 & 0 & 0                & 1                 &                               0 & 0 \\
0 & 0 & p_y \lambda^{-4} & -p_x \lambda^{-4} & (\sqrt{2}/2) \lambda^{-1/2}     & -(\sqrt{2}/2) \lambda^{-1/2} \\
0 & 0 & 0                & 0                 & (\sqrt{2}/2) \lambda^{1/2}      & (\sqrt{2}/2) \lambda^{1/2}
\end{pmatrix},
\]
and we left as an exercise to the reader to check that $C_0(x)^\top \Omega_3 C_0(x)=\Omega_3$,
where $\Omega_3$ is the matrix of the canonical symplectic form. 
The inverse of $C_0$ is given by
\[
C_0^{-1} =
\begin{pmatrix}
1 & 0 & 0 & 0 &0 &  -p_y \lambda^{-4} \\
0 & 1 & 0 & 0 &0 &   p_x \lambda^{-4} \\
0 & 0 & 1 & 0 &0 &  0 \\
0 & 0 & 0 & 1 &0 &  0 \\
0 & 0 &-(\sqrt{2}/2) p_y \lambda^{-7/2} & (\sqrt{2}/2) p_x \lambda^{-7/2} &  (\sqrt{2}/2) \lambda^{1/2} &  (\sqrt{2}/2) \lambda^{-1/2} \\
0 & 0 & (\sqrt{2}/2) p_y \lambda^{-7/2} & -(\sqrt{2}/2) p_x \lambda^{-7/2} & -(\sqrt{2}/2) \lambda^{1/2} & (\sqrt{2}/2) \lambda^{-1/2}
\end{pmatrix}\,,
\]
and it turns out that this frame allows us to reduce $\Dif X_0 \circ P_0$ as follows
\[
C_0^{-1} \Dif X_0 \circ P_0 C_0 =
\begin{pmatrix}
0 & 0 & p_y^2 \lambda^{-6} +1 & -p_xp_y \lambda^{-6} & 0 &0 \\
0 & 0 & -p_x p_y \lambda^{-6} & p_x^2 \lambda^{-6} +1 & 0 & 0 \\
0 & 0 & 0 & 0 & 0 & 0 \\
0 & 0 & 0 & 0 & 0 & 0 \\
0 & 0 & 0 & 0 & \lambda & 0 \\
0 & 0 & 0 & 0 & 0 & -\lambda
\end{pmatrix}.
\]
Notice that this expression corresponds to $C_0^{-1} \cR_0(C_0)$ since
in this case $C_0^{-1} \Dif (C_0)R_0=0$.
Finally, the reduced vector field is given by
$R_0 = \omega_1 \partial_x + \omega_2 \partial_y$,
where we recall that 
$\omega_1 = p_x (1+\lambda^{-2})$, and
$\omega_2 = p_y (1+\lambda^{-2})$.

To obtain the corrections $P_1$ and $R_1$ of the parameterization and the
reduced vector field, respectively, we consider the equation
\[
(\Dif X_{0} \circ P_0) P_1 - \Dif P_0 R_1 - \Dif P_1 R_0 = \eta
\]
where
\[
\eta =  - X_{1} \circ P_0 =
\begin{pmatrix}
\hat C \cos y \\
\hat B \sin x \\
- \hat B \omega_2 \cos x \\
\hat C \omega_1 \sin y \\
\hat C \sin y + \hat B \cos x \\
\hat C \cos z^* \cos y - \hat B \sin z^* \sin x
\end{pmatrix}.
\]
Using Lemma~\ref{lem:coho}, with $P_1=C_0 \hat \xi$ and $R_1=\rho$, we obtain the equivalent system of equations
\begin{align}
(p_y^2 \lambda^{-6}+1) \hat \xi_3 - p_xp_y \lambda^{-6} \hat \xi_4 - \Lie{R_0}(\hat \xi_1) = {} & \hat \eta_1 + \rho_1\,, \label{eq:abc:1:coho1} \\
-p_x p_y \lambda^{-6} \hat \xi_3 + (p_x^2 \lambda^{-6} +1) \hat \xi_4 - \Lie{R_0}(\hat \xi_2) = {} & \hat \eta_2 + \rho_2\,, \label{eq:abc:1:coho2} \\
- \Lie{R_0}(\hat \xi_3) = {} & \hat \eta_3 + \rho_3\,, \label{eq:abc:1:coho3} \\
- \Lie{R_0}(\hat \xi_4) = {} & \hat \eta_4 + \rho_4\,, \label{eq:abc:1:coho4} \\
\lambda \hat \xi_5 - \Lie{R_0}(\hat \xi_5) = {} & \hat \eta_5\,, \label{eq:abc:1:coho5}\\
-\lambda \hat \xi_6- \Lie{R_0}(\hat \xi_6) = {} & \hat \eta_6\,, \label{eq:abc:1:coho6}
\end{align}
where $\Lie{R_0}(\hat \xi_i)=\omega_1 \partial_x \hat \xi_i + \omega_2 \partial_y \hat \xi_i$, and
\[
\hat \eta = C_0^{-1} \eta =
\begin{pmatrix}
\hat \eta_1 \\
\hat \eta_2 \\
\hat \eta_3 \\
\hat \eta_4 \\
\hat \eta_5 \\
\hat \eta_6 \\
\end{pmatrix}
=
\begin{pmatrix}
A_1 \cos y + A_2 \sin x \\
A_3 \cos y + A_4 \sin x \\
A_5 \cos x \\
A_6 \sin y \\
A_7 \cos x + A_8 \cos y + A_9 \sin x + A_{10} \sin y \\
A_{11}\cos x + A_{12} \cos y + A_{13} \sin x + A_{14} \sin y \\
\end{pmatrix}.
\]
The coefficients $A_i$, $i=1,\ldots 14$, are functions depending on the action variables $p_x,p_y$,
given by
{\allowdisplaybreaks
\begin{equation}\label{eq:coefs:A}
\begin{split}
A_1 = {} & \vphantom{\sqrt{2}{2}} \hat C(1+p_y^2 \lambda^{-6}) \\
A_2 = {} & \vphantom{\sqrt{2}{2}} - \hat B p_x p_y \lambda^{-6} \\
A_3 = {} & \vphantom{\sqrt{2}{2}} - \hat C p_x p_y \lambda^{-6} \\
A_4 = {} &\vphantom{\sqrt{2}{2}}  \hat B(1+p_x^2 \lambda^{-6}) \\
A_5 = {} & \vphantom{\sqrt{2}{2}} - \hat B \omega_2 \\
A_6 = {} & \vphantom{\sqrt{2}{2}} \hat C \omega_1 \\
A_7 = {} & \sqrt{2}/2 (\lambda^6 + \lambda^2 p_y^2 + p_y^2) \hat B \lambda^{-11/2}
\end{split}
\qquad
\qquad
\begin{split}
A_8 = {} & -\sqrt{2}/2 \hat C p_y \lambda^{-5/2} \\
A_9 = {} & \sqrt{2}/2 \hat B p_x \lambda^{-5/2} \\
A_{10}={}& \sqrt{2}/2 (\lambda^6 + \lambda^2 p_x^2 + p_x^2) \hat C \lambda^{-11/2} \\
A_{11}={}& -\sqrt{2}/2 (\lambda^6 + \lambda^2 p_y^2 + p_y^2) \hat B \lambda^{-11/2} \\
A_{12}={}& -\sqrt{2}/2 \hat C p_y \lambda^{-5/2} \\
A_{13}={}& \sqrt{2}/2 \hat B p_x \lambda^{-5/2} \\
A_{14}={}& -\sqrt{2}/2 (\lambda^6 + \lambda^2 p_x^2 + p_x^2) \hat C \lambda^{-11/2} 
\end{split}
\end{equation}
}

The solution of Eqs.~\eqref{eq:abc:1:coho5} and \eqref{eq:abc:1:coho6}, using Fourier series, is obtained using~\eqref{eq:solc:coho:general}
\begin{align*}
\hat \xi_5 = {} & B_1 \cos x + B_2 \cos y + B_3 \sin x + B_4 \sin y\,, \\
\hat \xi_6 = {} & B_1 \cos x - B_2 \cos y - B_3 \sin x + B_4 \sin y\,,
\end{align*}
where the coefficients $B_i$ have the following expressions:
\[
B_1 = \frac{A_7 \lambda+\omega_1 A_9}{\lambda^2 +\omega_1^2}\,, \quad
B_2 = \frac{A_8 \lambda+\omega_2 A_{10}}{\lambda^2 +\omega_2^2}\,, \quad
B_3 = \frac{A_9 \lambda-\omega_1 A_7}{\lambda^2 +\omega_1^2}\,, \quad
B_4 = \frac{A_{10} \lambda-\omega_2 A_8}{\lambda^2 +\omega_2^2}\,,
\]
which are functions depending on the action variables $p_x,p_y$. The solution of Eqs.~\eqref{eq:abc:1:coho1}--~\eqref{eq:abc:1:coho4} is given by
$\hat \xi_1=\hat \xi_2=\hat \xi_3=\hat \xi_4=0$ and
\begin{equation}\label{eq:ham:R1}
R_1 = \rho = 
\begin{pmatrix}
-A_1 \cos y - A_2 \sin x \\
-A_3 \cos y - A_4 \sin x \\
-A_5 \cos x \\
-A_6 \sin y 
\end{pmatrix}\,.
\end{equation}
By construction, the vector field $R_1$ in~\eqref{eq:ham:R1} is Hamiltonian
with respect to the symplectic form $dp_x \wedge dx + dp_y \wedge dy$ (c.f. Section~\ref{ssec:NHIM:symp}).

The specific computations regarding $C_1$, $R_2$ and $P_2$ are omitted, since they will not
be used in what follows. The only thing that we need to know in the next section is which resonances appear
in the averaging process of the Hamiltonian corresponding to $R_2$.
We remark that in our problem, it turns out that $R_2$ is a trigonometric polynomial of degree $2$. This claim follows from the the construction explained in Section~\ref{ssec:NHIM:symp} and the fact that we know the degrees of the functions $X_0$, $X_1$, $X_2$, $P_0$, $P_1$ and $C_0$ that appear in Eqs.~\eqref{eq:xi2:corre} and~\eqref{tonto}.

\subsection{Invariant tori on the NHIM}\label{ssec:inner:tori}

From the computations presented in Section~\ref{ssec:NHIM:ABC},
we obtain that the dynamics reduced to
the perturbed NHIM is given by the Hamiltonian system:
\begin{equation}\label{eq:red:ham}
r_\ep(x,y,p_x,p_y)=r_0(p_x,p_y) + r_1(x,y,p_x,p_y) \ep + r_2(x,y,p_x,p_y) \ep^2 + \cO(\ep^3)\,.
\end{equation}
The Hamiltonian functions $r_i$ satisfy $R_i = -\Omega_2^0 \Dif r_i^\top$, where $R_i$ is
the reduced vector field on the NHIM computed in Section~\ref{ssec:NHIM:ABC}. Specifically,
we have
\begin{align*}
r_0(p_x,p_y)={} & \frac{p_x^2+p_y^2}{2}+\sqrt{p_x^2+p_y^2}\,,\\
r_1(x,y,p_x,p_y)={} & A_5 \sin x - A_6 \cos y\,,\\
r_2(x,y,p_x,p_y)={} & A_{15} + A_{16} \cos x + A_{17} \cos y + A_{18} \sin x + A_{19} \sin y \\
                    & + A_{20} \cos (2x) + A_{21} \cos(2y) + A_{22} \cos(x+y) + A_{23} \cos(x-y) \\
                    & + A_{24} \sin (2x) + A_{25} \sin(2y) + A_{26} \sin(x+y) + A_{27} \sin(x-y)\,,
\end{align*}
where $A_5$ and $A_6$ are given in Eqs.~\eqref{eq:coefs:A} and the remaining coefficients
are certain explicit functions of $(p_x,p_y)$ whose explicit expressions are not important in the computations performed later. Here we are denoting as $(x,y,p_x,p_y)$, with abuse of notation, the reduced variables on the perturbed NHIM, but they are not the same as the coordinate variables in the phase space $\TT^3\times\RR^3$. However, at first order in $\ep$ the parameterization is a graph (see Eq.~\eqref{eq:defP1}), and hence, the reduced variables and the coordinate variables only differ in terms of order $\ep^2$.

\subsubsection{The global averaging method}\label{sssec:aver}

The task now is to 
characterize invariant tori on the perturbed NHIM.
The idea introduced in~\cite{DLS06} consists in performing several steps of averaging 
in a global way on the whole NHIM. To this
end, a normal form procedure is applied but, when we are close to a given resonance,
the resonant normal form is defined by evaluating the corresponding coefficient on the
resonant manifold (see also~\cite{DLS13}).
It is worth mentioning that since the problems considered in~\cite{DLS06,DLS13} are
non-autonomous, a suitable projection is the so-called projection along the $k$-direction.
In our case, due to the fact that the studied Hamiltonian is autonomous, the
orthogonal projection is more appropriate to perform computations.

Although we are interested in the ABC system, 
the discussion of this section is presented in a general setting.
This allows us to use a more convenient notation and, moreover,
we think that it will 
help
the reader to link with the exposition in~\cite{DLS06,DLS13} and to consult the
details that we omit in our discussion.

Let us consider a Hamiltonian system on $N=\TT^n \times \RR^n$ of the form
\begin{equation}\label{eq:aver:gen}
h(u,p_u) = h_0(p_u) + \sum_{i=1}^m \ep^i h_i(u,p_u) + \cO(\ep^{m+1})\,,
\end{equation}
where every $h_i$ is written in Fourier series as
\begin{equation}\label{eq:Fou:aver}
h_i(u,p_u)= \sum_{k\in \cZ_i \subset \ZZ^n} \left( h^{\cos}_{i,k}(p_u) \cos(k \cdot u) + h^{\sin}_{i,k}(p_u) \sin(k\cdot u) \right)\,,
\end{equation}
where $\cZ_i$ is the support of the Fourier series, which is assumed to
be a finite set. For the sake of consistency we take $h^{\sin}_{i,0}\equiv0$.

The averaging of Eq.~\eqref{eq:aver:gen}
consists in performing (recursively) a suitable change of variables in such a way
that we obtain a new Hamiltonian system depending on the variables $u\in \TT^n$
in a simple way. Setting $\langle h\rangle_{0}(u,p_u):=h(u,p_u)$, let us assume that we have performed $m-1\geq0$ steps of
averaging, so we have
\[
\langle h \rangle_{m-1} (u,p_u) = h_0 (p_u)+\sum_{i=1}^{m-1} \ep^i \bar h_i(u,p_u) + \ep^m  h_m(u,p_u) + \cO(\ep^{m+1})\,.
\]
Then, given a Hamiltonian system $\ep^mg_m$ with time-1 flow $\phi^{g_m}$, we introduce
the new Hamiltonian
\begin{align*}
\langle h \rangle_{m} (u,p_u) = {} & \langle h \rangle_{m-1} \circ \phi^{g_m} (u,p_u) \\
= {} & h_0 (p_u)+\sum_{i=1}^{m-1} \ep^k \bar h_i(u,p_u) +
\ep^m 
\bigg(
h_m(u,p_u)+\{ h_0,g_m \}(u,p_u)
\bigg) + \cO(\ep^{m+1})\,,
\end{align*}
and we ask it
to be as simple as possible by taking
\[
h_m(u,p_u)+\{ h_0,g_m \}(u,p_u) = \bar h_m(u,p_u)\,.
\]
Here $\{\cdot,\cdot\}$ is the Poisson bracket, defined as
\[
\{f,g\}=\sum_{i=1}^n 
\left(
\frac{\partial f}{\partial u_i} \frac{\partial g}{\partial p_{u,i}}
-
\frac{\partial f}{\partial p_{u,i}} \frac{\partial g}{\partial u_{i}}
\right)\,.
\]
Using an expansion as in~\eqref{eq:Fou:aver}, we obtain the following set of equations
for the Fourier coefficients:
\begin{equation}\label{eq:coefs:Fou}
\begin{split}
(\omega \cdot k) g^{\cos}_{m,k} (p_u) = {} & \bar h^{\sin}_{m,k} (p_u)- h^{\sin}_{m,k}(p_u)\,, \\
-(\omega \cdot k) g^{\sin}_{m,k} (p_u)= {} & \bar h^{\cos}_{m,k} (p_u)- h^{\cos}_{m,k}(p_u)\,, 
\end{split}
\end{equation}
for every $k \in \ZZ^n \backslash\{0\}$, and we take $\bar h^{\cos}_{m,0} = h^{\cos}_{m,0}$,
and $\bar h^{\sin}_{m,0}=0$, so that $g^{\cos}_{m,0}$ and $g^{\sin}_{m,0}$ can
take any value. In these equations $\omega\equiv \omega(p_u):=\frac{\partial h_0}{\partial p_u}$.

\begin{definition}
Given a Hamiltonian $h: \cI \subset \RR^n \rightarrow \RR$, for each $k \in \ZZ^n \backslash\{0\}$ we define the resonant set
\[
\mathrm{R}_k:=\{ p_u \in \cI \, : \, \omega(p_u)\cdot k=0\}.
\]
\end{definition}

Let us assume in what follows that the function $\omega(p_u)\cdot k$ has no critical points on $\mathrm{R}_k$, so that the resonant manifolds are smooth surfaces (a condition that
depends only on the unperturbed problem and that is certainly satisfied
by the ABC system). Then, it makes perfect sense to introduce some
additional definitions and notation. Indeed,
given a resonant set $\mathrm{R}_k$ and a small constant $L>0$, we
denote the tubular neighborhood of $\mathrm{R}_k$ of radius $L$ (measured with the Euclidean metric) as $\mathrm{Tub}(\mathrm{R}_k,L)$.
Moreover, for every resonant set we introduce the orthogonal projection
$\Pi_{k} : \mathrm{Tub}(\mathrm{R}_k,L) \subset \RR^n \rightarrow \mathrm{R}_k$.
Finally, given a resonant set $\mathrm{R}_k$ we denote by $\mathrm{dist}\,(p_u,\mathrm{R}_k)$
the Euclidean distance of the point $p_u$ to the manifold $\mathrm{R}_k$.

Notice that $\mathrm{R}_k=\mathrm{R}_{mk}$ for any $m\in \ZZ$. Then, given two
sets $\mathrm{R}_k$ and $\mathrm{R}_\ell$, we have, generically, the following trichotomy:
\begin{itemize}
\item They are the same manifold: $\mathrm{R}_k=\mathrm{R}_\ell$, i.e, $k=m\ell$ for some $m\in \ZZ$.
\item They do not intersect: $\mathrm{R}_k \cap \mathrm{R}_\ell = \varnothing$.
\item They intersect transversely in a manifold of codimension two without boundary.
\end{itemize}

It is worth mentioning that the third case does not play an important role in our
problem. Indeed, for the ABC system, resonant sets are 1-dimensional manifolds and their intersections define sets of dimension zero. The case of higher dimensions
has been discussed recently in~\cite{DLS13} proving that the existence of
multiple resonances is not a limitation to prove diffusion.

If there is a finite number of resonant sets, it is clear that we can choose
a constant $L>0$ small enough such that for every pair $k,\ell \in \ZZ^n$ we
have either $\mathrm{R}_k=\mathrm{R}_\ell$ or 
$\mathrm{R}_k \cap \mathrm{Tub}(\mathrm{R}_\ell,L) = \varnothing$.
Then, following~\cite{DLS13}, we construct a solution of Eq.~\eqref{eq:coefs:Fou} in a global way,
that is, for all values $p_u \in \cI$. Of course, we only want to modify
the Fourier coefficients in the support of the series, that is,
if $k \notin \cZ_m$ we take $\bar h^{\cos}_{m,k}(p_u)=\bar h^{\sin}_{m,k}(p_u)=0$,
and hence $g^{\cos}_{m,k}(p_u)=g^{\sin}_{m,k}(p_u)=0$.
Then, if
$k \in \cZ_m$, we take
\begin{align*}
\bar h^{\cos}_{m,k}(p_u) = & h^{\cos}_{m,k}(\Pi_k(p_u)) \psi \bigg(\frac{1}{L} \mathrm{dist}\,(p_u,\mathrm{R}_k)\bigg)\,,\\
\bar h^{\sin}_{m,k}(p_u) = & h^{\sin}_{m,k}(\Pi_k(p_u)) \psi \bigg(\frac{1}{L} \mathrm{dist}\,(p_u,\mathrm{R}_k)\bigg)\,,
\end{align*}
where $\psi :\RR \rightarrow \RR$ is a fixed $C^\infty$ function such that $\psi(t)=1$, if $t \in [-1,1]$, and $\psi(t)=0$, if $t\notin [-2,2]$. The Fourier coefficients of the Hamiltonian
$g_m$ are obtained from Eq.~\eqref{eq:coefs:Fou}, passing to the limit when $p_u$ tends to $\mathrm{R}_k$. For details, we refer to Lemma 8.8 in~\cite{DLS06} and to Lemma 10 in~\cite{DLS13}. With this choice we distinguish two different zones:

\begin{itemize}
\item \textbf{Non-resonant region:} 
If $p_u \notin \mathrm{Tub}(\mathrm{R}_k,2L)$,
we have $\bar h^{\cos}_{m,k}(p_u)=0 = \bar h^{\sin}_{m,k}(p_u)=0$.

\item \textbf{Resonant region:} 
If $p_u \in \mathrm{Tub}(\mathrm{R}_k,L)$, we have $\bar h^{\cos}_{m,k}(p_u) = h^{\cos}_{m,k}(\Pi_k(p_u))$,
and $\bar h^{\sin}_{m,k}(p_u) = h^{\sin}_{m,k}(\Pi_k(p_u))$.
\end{itemize}

\begin{remark}\label{rem:role:L}
The choice of $L$ is arbitrary. This implies that
we do not need to study the regions at a distance between $L$ and $2L$ of the resonant set $\mathrm{R}_k$.
\end{remark}

\subsubsection{Adapted coordinates on the averaged system}\label{ssec:aver:ABC}

Let us now apply the averaging procedure described in Section~\ref{sssec:aver}
to the reduced Hamiltonian~\eqref{eq:red:ham}. In this case, resonant
sets are expressed as
\[
\mathrm{R}_k = \{ (p_x,p_y) \in \cI \, : \, \omega_1 k_1 + \omega_2 k_2 = 0 \}
\]
where $k=(k_1,k_2)$, the set $\cI\subset \{p_x>0\}\times\{p_y>0\}$, and the frequency $\omega=(\omega_1,\omega_2)$
is given by Eqs.~\eqref{eq:omega1} and~\eqref{eq:omega2}. Then, it is clear
that there are no resonances associated to the averaging of order $|k|\leq 1$, since
$p_x\neq 0$ and $p_y \neq 0$, and so we have $\omega_1 \neq 0$ and $\omega_2 \neq 0$.
For the same reason, in the averaging of order $|k|\leq 2$, we must deal only with the set
$\omega_1-\omega_2=0$, that corresponds to the 
straight line $p_x=p_y$. The orthogonal projection associated to this particular
resonance, that we simply write as $\Pi$, has the following explicit expression:
\begin{equation}\label{eq:Gamma}
\Pi(p_x,p_y)=\left( \frac{p_x+p_y}{2}, \frac{p_x+p_y}{2} \right)\,.
\end{equation}

\begin{itemize}
\item \textbf{Non-resonant region:} we can eliminate all the terms in $r_1(u,p_u)$ and $r_2(u,p_u)$, so
the second order averaged system is given by
\[
\langle r_\ep \rangle_2(x, y,p_x,p_y) = r_0(p_x,p_y) + \cO(\ep^3).
\]
Neglecting the $\cO(\ep^3)$ terms, we obtain an integrable unperturbed system. The invariant
tori of this unperturbed system are given by the level sets
\begin{equation}\label{eq:nueva}
\begin{split}
p_x = {} & e_1\,,\\
p_y = {} & e_2\,.
\end{split}
\end{equation}
When the perturbation $\cO(\ep^3)$ is taken into account, KAM theorem guarantees
that most of these invariant tori persist for the perturbed system, covering the non-resonant region except for a set of measure of order $\cO(\ep^{3/2})$. 
We remark again that, since our problem is real analytic, we do not need to
care about the technical difficulties regarding regularity in the KAM theorem.

\item \textbf{Resonant region:} we consider the projection~\eqref{eq:Gamma},
and we obtain that the second order averaged reduced Hamiltonian is given by
\[
\langle r_\ep \rangle_2 (x, y,p_x,p_y) = r_0(p_x,p_y) + \ep^2 \bigg( A_{23}(\Pi(p_x,p_y)) \cos(x-y) +  A_{27}(\Pi(p_x,p_y)) \sin(x-y) \bigg) + \cO(\ep^3)\,.
\]
Then, it is natural to perform a canonical change of variables in order to introduce a resonant angle: 
\[
\begin{split}
\theta_1 = {} & x, \\
\theta_2 = {} & x - y,
\end{split}
\qquad 
\qquad
\begin{split}
I_1 = {} & p_x+p_y, \\
I_2 = {} & -p_y,
\end{split}
\]
thus obtaining the Hamiltonian
\begin{equation}\label{eq:aver1}
\langle r_\ep \rangle_2 (\theta_1,\theta_2, I_1,I_2) 
= 
r_0(I_1+I_2,-I_2) 
+ \ep^2 \bigg( A_{23}(\tfrac{I_1}{2},\tfrac{I_1}{2}) \cos(\theta_2) 
+  A_{27}(\tfrac{I_1}{2},\tfrac{I_1}{2}) \sin(\theta_2) \bigg) + \cO(\ep^3)\,.
\end{equation}

The next step is to perform a Taylor expansion around the resonance. It is
clear that the resonance $p_x=p_y$ is equivalent to $I_2=-I_1/2$. Hence,
we consider the expansion $I_2=-I_1/2+\delta$ and we write the unperturbed Hamiltonian as follows
\begin{align*}
r_0(\tfrac{I_1}{2}+ \delta,\tfrac{I_1}{2}-\delta) = {} & r_0(\tfrac{I_1}{2},\tfrac{I_1}{2}) 
+ \frac{1}{2} \bigg (
\frac{\partial^2 r_0}{\partial p_x^2}
-2\frac{\partial^2 r_0}{\partial p_x \partial p_y}
+ \frac{\partial^2 r_0}{\partial p_y^2}
\bigg) (\tfrac{I_1}{2},\tfrac{I_1}{2}) \delta^2 + \cO(\delta^3),
\end{align*}
where we have used that $\omega_1=\omega_2$ at $\delta=0$. Moreover, using the
specific expression of $r_0$, it turns out that we can write the
Hamiltonian~\eqref{eq:aver1} as
\begin{equation}\label{eq:aver:trunc}
\langle r_\ep \rangle_2 (\theta_1,\theta_2, I_1,-\frac{I_1}{2}+\delta) =
\frac{I_1^2}{4}+\frac{I_1}{\sqrt{2}} + (1+\sqrt{2} I_1^{-1}) \delta^2 
+ \ep^2 \bigg( A_{23}(\tfrac{I_1}{2},\tfrac{I_1}{2}) \cos(\theta_2) 
+  A_{27}(\tfrac{I_1}{2},\tfrac{I_1}{2}) \sin(\theta_2) \bigg)\,,
\end{equation}
modulo terms of order $\cO(\ep^3,\delta^3)$.
This corresponds to a pendulum-like Hamiltonian system
in the variables $(\theta_2,\delta)$ depending on the variable $I_1$. In other words, we observe that $I_1$ is an integral of
motion of the truncated Hamiltonian~\eqref{eq:aver:trunc} and the motion of the 
variables $(\theta_2,\delta)$ is described by the system
\begin{align*}
\dot{\theta_2} = {} & 2 (1+\sqrt{2} I_1^{-1}) \delta ,\\
\dot{\delta} = {} & \ep^2  \bigg( A_{23}(\tfrac{I_1}{2},\tfrac{I_1}{2}) \sin(\theta_2) 
-  A_{27}(\tfrac{I_1}{2},\tfrac{I_1}{2}) \cos(\theta_2) \bigg)
\end{align*}
The above system has a hyperbolic equilibrium point at $(\theta_2,\delta)=(\theta_2^*(I_1),0)$,
and we denote by $H^*$ the energy of this point. Then, the level sets of $I_1$ and $\langle r_\ep\rangle_2$ are characterized
as follows
\begin{equation}\label{eq:unper:torus:res}
\begin{split}
I_1 = {} & e_1\,,\\
\frac{I_1^2}{4}+\frac{I_1}{\sqrt{2}} + (1+\sqrt{2} I_1^{-1}) \delta^2 
+ \ep^2 \bigg( A_{23}(\tfrac{I_1}{2},\tfrac{I_1}{2}) \cos(\theta_2) 
+  A_{27}(\tfrac{I_1}{2},\tfrac{I_1}{2}) \sin(\theta_2) \bigg) +\cO(\ep^3,\delta^3)= {} & e_2\,.
\end{split}
\end{equation}
We observe that
these level sets have different topology
depending if $e_2>H^*$ (two primary tori), if
$e_2=H^*$ (two whiskered tori with coincident whiskers) or if $e_2<H^*$ (two secondary tori).
Again, applying the KAM theorem to consider the effect of the perturbation, we
obtain that many of the invariant tori in the previous picture persist, covering
the resonant region except for a set of measure $\cO(\ep^{3/2})$. We refer to~\cite{DLS06,DLS13,DH09}
for full details on the application of the KAM theorem close to the separatrix.
As before, we do not discuss here the specific technical details since they are
covered by the fact that our Hamiltonian is real-analytic.
\end{itemize}

We have obtained an approximation of the level sets that characterize
the invariant objects inside the NHIM. Such level sets are not written in
terms of the original variables of the problem but in the averaged variables.
In the following result we translate the previous construction into the coordinate
variables in phase space.

\begin{proposition}\label{prop:level:sets}
Let us consider the original Hamiltonian system
\[
H_\ep= H_0 + \ep H_1 + \ep^2 H_2,
\]
where $H_0$, $H_1$ and $H_2$ are given by Eqs.~\eqref{eq:H0:two}--~\eqref{eq:H2}. Then,
the invariant tori inside the NHIM are characterized
by the level sets
\begin{equation}\label{eq:tori:non-resonant}
\begin{split}
p_x - \ep \hat B \tfrac{\omega_2}{\omega_1} \sin x + \cO(\ep^2) = {} & e_1\,, \\
p_y - \ep \hat C \tfrac{\omega_1}{\omega_2} \cos y + \cO(\ep^2) = {} & e_2\,, 
\end{split}
\end{equation}
in the non-resonant region, and
\begin{equation}\label{eq:tori:resonant}
\begin{split}
p_x + p_y - \ep \hat B \tfrac{\omega_2}{\omega_1} \sin x - \ep \hat C \tfrac{\omega_1}{\omega_2} \cos y + \cO(\ep^2) = {} & e_1, \\
\frac{e_1^2}{4}+\frac{e_1}{\sqrt{2}} + (1+\sqrt{2} e_1^{-1}) \left(\tfrac{e_1}{2}-p_y + \ep \hat C \tfrac{\omega_1}{\omega_2} \cos y\right)^2 + \cO(\ep^2) = {} & e_2, 
\end{split}
\end{equation}
in the resonant region. Recall that the frequencies $\omega_1$ and $\omega_2$ are defined in Eqs.~\eqref{eq:omega1} and~\eqref{eq:omega2}.
\end{proposition}

\begin{proof}
We only have to undo the different changes of variables in the averaging construction previously explained. In particular
we recall that we defined
\[
\langle r_\ep \rangle_1 (u,p_u) = r_\ep \circ \phi^{g_1}(u,p_u),
\]
where $\phi^{g_1}$ is the time-$1$ flow of a Hamiltonian $\ep g_1$ satisfying $r_1+\{ r_0,g_1\}=0$, with $r_1
= A_5 \sin x - A_6 \cos y$. The expressions of
$A_5$ and $A_6$ are given in Eqs.~\eqref{eq:coefs:A}. Since there are no resonances involved,
we can solve the previous equation taking
\[
g_1(u,p_u)=G_1 \cos x + G_2 \sin y,
\]
with $G_1= \hat B \tfrac{\omega_2}{\omega_1}$ and $G_2= - \hat C \tfrac{\omega_1}{\omega_2}$. This means
that we have to invert the change of variables
\begin{equation}\label{eq:change:g}
\begin{split}
x \mapsto {} & x + \ep \partial_{p_x} g_1 + \cO(\ep^2),\\
y \mapsto {} & y + \ep \partial_{p_y} g_1 + \cO(\ep^2),\\
p_x \mapsto {} & p_x - \ep \partial_{x} g_1 + \cO(\ep^2),\\
p_y \mapsto {} & p_y - \ep \partial_{y} g_1 + \cO(\ep^2),
\end{split}
\end{equation}
that lead to the averaged system.
As explained at the beginning of Section~\ref{ssec:inner:tori}, the reduced variables are the same as the coordinate variables up to terms of order $\ep^2$, so we can safely assume that the variables $(x,y,p_x,p_y)$ in Eq.~\eqref{eq:change:g} are the phase space coordinates. 

Let us first consider the non-resonant region. The unperturbed invariant tori
of the averaged system are given by the level sets $p_x=e_1$ and $p_y=e_2$, cf. Eq.~\eqref{eq:nueva}.
Inverting the change of variables~\eqref{eq:change:g} we obtain that
the surviving invariant tori satisfy the expression in Eq.~\eqref{eq:tori:non-resonant}.

In the resonant region, we obtained that the unperturbed invariant
tori of the averaged system are given by the level sets in Eq.~\eqref{eq:unper:torus:res}.
Following~\cite{DLS13}, we first replace the expression $I_1=e_1$ into
the second expression in~\eqref{eq:unper:torus:res}, thus obtaining the equivalent system
\begin{equation}\label{eq:unper:torus:res2}
\begin{split}
I_1 = {} & e_1\,,\\
\frac{e_1^2}{4}+\frac{e_1}{\sqrt{2}} + (1+\sqrt{2} e_1^{-1}) \delta^2 
+ \cO(\ep^2) = {} & e_2\,.
\end{split}
\end{equation}
This choice will simplify subsequent computations. Then, we recall the
definition of the variable $\delta$ and we invert the change of variables
$(p_x,p_y) \mapsto (I_1,I_2)$, thus obtaining
\[
\delta= I_2 + \frac{I_1}{2}=-p_y + \frac{e_1}{2}.
\]
Then, inverting the change of variables~\eqref{eq:change:g}, we
obtain Eq.~\eqref{eq:tori:resonant}.
\end{proof}
We would like to remark that the terms of order $\ep$ in Eqs.~\eqref{eq:tori:non-resonant} and~\eqref{eq:tori:resonant} will be important in the computations of Section~\ref{sec:chains}. These terms
are not necessary in~\cite{DLS06,DLS13,DH09}, due to the fact that the unperturbed outer dynamics
is the identity and hence there is no phase-shift. 

\section{Outer dynamics of the NHIM}
\label{sec:outer}

In this section we consider the outer dynamics of the NHIM for the perturbed system.
This dynamics is modelled by the so-called \emph{scattering map} of a
normally hyperbolic invariant manifold with intersecting stable and unstable
invariant sets along a homoclinic manifold. This remarkable
tool was introduced in~\cite{DLS00} to study Arnold diffusion in the context
of periodic perturbations of geodesic flows in $\TT^2$,
and it was crucial for applications in~\cite{DLS06,DLS13,DH11,GL06}.
The paper~\cite{DLS08} contains a complete description of the geometric properties
of the scattering map, together with a systematic development of perturbative
formulas for its computation. 

In Section~\ref{ssec:poincare} we recall the construction of the so-called Melnikov
potential, which was introduced in~\cite{DLS00}, in the setting considered in this
paper. In Section~\ref{ssec:scattering} we present a brief definition of the scattering
map and we obtain its first order approximation for the case of the ABC system.

\subsection{The Poincar\'e-Melnikov function}
\label{ssec:poincare}

As was discussed in Section~\ref{sec:inner}, for small values of $\ep$ there exists a perturbed NHIM, denoted
by $\Lambda_\ep$, together with local invariant manifolds 
$W^\srm_{\mathrm{loc}}(\Lambda_\ep)$ and $W^\urm_{\mathrm{loc}}(\Lambda_\ep)$.
These manifolds are $\cO(\ep)$-close to $\Lambda$ and 
$W^\srm(\Lambda)=W^\urm(\Lambda)$, respectively. As usual, we globalize the invariant manifolds as
$
W^\srm(\Lambda_\ep) = \bigcup_{t<0} \phi_t^\ep (W^\srm_{\mathrm{loc}}(\Lambda_\ep))$,
$W^\urm(\Lambda_\ep) = \bigcup_{t>0} \phi_t^\ep (W^\urm_{\mathrm{loc}}(\Lambda_\ep))$,
where $\phi_t^\ep$ is the flow of the perturbed Hamiltonian $H_\ep$. The intersections of the stable and unstable invariant manifolds are given by the following proposition. All along this section we use the notation introduced in Section~\ref{ssec:geom:unper} for the unperturbed problem.

\begin{proposition}\label{prop:melnikov}
Let us consider an analytic Hamiltonian system of the form
$
H_\ep(q,p) = H_0(q,p) + \ep h(q,p,\ep),
$
having a NHIM $\Lambda_\ep$, where the unperturbed Hamiltonian
$H_0$ is given by~\eqref{eq:sys:H0}.
The homoclinic intersections
of the invariant manifolds $W^\srm(\Lambda_\ep)$ and $W^\urm (\Lambda_\ep)$ are
described, at first order in $\ep$, by the critical
points of the \emph{Poincar\'e function} (also known as \emph{Melnikov potential}):
\begin{equation}\label{eq:Poin:Mel}
L(\tau,x,y,p_x,p_y)= \int_{-\infty}^{\infty} h(\phi_\sigma^0(u^0),0)-h(\phi_{\sigma}^0 (u^* + u_{\pm}),0) d \sigma\,,
\end{equation}
where $\phi_\sigma^0$ is the time-$\sigma$ flow of the unperturbed Hamiltonian $H_0$.
In particular, $\phi_\sigma^0(u^*+u_{\pm})$ is given by Eqs.~\eqref{eq:flow:H0:inner} and~\eqref{eq:phase-shift},
and $\phi_\sigma^0(u^0)$ is given by Eq.~\eqref{eq:flow:H0:outer}.
Recall that the compact notation $u_{\pm}$ means that 
we take $u_+$ for $\sigma>0$ and $u_-$ for $\sigma<0$.
\end{proposition}

We observe that this expression of $L(\tau,x,y,p_x,p_y)$ differs from the one used
in~\cite{DLS06,DLS13,DH09} by the fact that it depends on the phase-shift.
A Melnikov potential of this type is given in Proposition 3 of~\cite{Tre02} and analogous expressions
can be found in~\cite{DG00,DLS00}. We invite the reader to compare Proposition~\ref{prop:melnikov}
with Theorem 32 in~\cite{DLS08} that is stated in a more general setting.
For the sake of completeness, we present here a complete proof of this proposition that may be of valuable help for the general reader.
The arguments, which we adapt from~\cite{DLS06}, are 
standard in Melnikov theory and well known to experts.

Let us consider the function
\begin{equation}\label{eq:Ham:P}
\cP(x,y,z,p_x,p_y,p_z):=
\frac{p_z^2}{2} - \lambda^2(\cos(z -\alpha)+1)\,, \\
\end{equation}
which is a first integral of the Hamiltonian system defined by $H_0$, where $\alpha=\arctan (p_x/p_y)$.
This function is used to estimate the distance between the invariant manifolds
associated to the NHIM (see Lemma~\ref{lem:heteroclinic} below). 
Indeed, at every point $u^0=u^0(\tau,x,y) \in W^\srm(\Lambda_0)=W^\urm(\Lambda_0)$,
given by~\eqref{eq:param:u0}, we have
\begin{align*}
\cP(u^0) = {} & \frac{2\lambda^2}{\cosh^2(\lambda \tau)} - \lambda^2 \left(\cos(4 \arctan \rme^{\lambda \tau}+\pi)+1 \right) \\
= {} & \lambda^2 \left( \frac{8\rme^{2 \lambda \tau}}{(1+\rme^{2\lambda \tau})^2} + \cos(4 \arctan \rme^{\lambda \tau})  - 1 \right) = 0\,.
\end{align*}
Then, for every point $u^0 \in W^\srm(\Lambda_0)=W^\urm(\Lambda_0)$ we consider the straight line $\Sigma$, transversal to $W^\srm(\Lambda_0)=W^\urm(\Lambda_0)$, given by
\[
\Sigma \equiv \Sigma(u^0) = \{u^0 + \mu \nabla_{(z,p_z)} \cP(u^0):\mu\in\RR\}\,,
\]
where we are using the notation 
$\nabla_{(z,p_z)}\cP:=(0,0,\frac{\partial \cP}{\partial z},0,0,\frac{\partial \cP}{\partial p_z})$. 
We denote by $u^\srm = \Sigma(u^0) \cap W^\srm(\Lambda_\ep)$ and $u^\urm = \Sigma(u^0) \cap W^\urm(\Lambda_\ep)$
the intersections of the line $\Sigma$ with the stable and unstable manifolds of $\Lambda_\ep$, respectively.
Then, there exist constants $\mu^\srm\in\RR$ and $\mu^\urm\in\RR$ such that these intersections are given by
\begin{align*}
u^\srm = {} & \bigg(x+F_1(\tau),y+F_2(\tau),z^0(\tau)+\mu^\srm \partial_z \cP^0,p_x,p_y,p_z^0(\tau) + \mu^\srm \partial_{p_z} \cP^0 \bigg)\,,\\
u^\urm = {} & \bigg(x+F_1(\tau),y+F_2(\tau),z^0(\tau)+\mu^\urm \partial_z \cP^0,p_x,p_y,p_z^0(\tau) + \mu^\urm \partial_{p_z} \cP^0 \bigg)\,,
\end{align*}
where $\partial_z \cP^0 := \frac{\partial \cP}{\partial z}(u^0(\tau,x,y))$
and $\partial_{p_z} \cP^0 = \frac{\partial \cP}{\partial p_z}(u^0(\tau,x,y))$. Then, the following
result states the relationship between $\cP(u)$
 and the intersections of $W^\srm(\Lambda_\ep)$ and $W^\urm(\Lambda_\ep)$ for $\ep\neq 0$:

\begin{lemma}\label{lem:heteroclinic}
For each fixed $u^0$, the homoclinic intersections of the stable and unstable manifolds
are characterized by 
\[
u^\srm = u^\urm
\quad
\Leftrightarrow
\quad
\mu^\srm=\mu^\urm
\quad
\Leftrightarrow
\quad
\cP(u^\srm)=\cP(u^\urm).
\]
\end{lemma}

\begin{proof}
It is clear that these implications hold from the left to the right. The converse follows from the fact that the function
\[
f(\mu):=\cP(u^0+\mu \nabla_{(z,p_z)} \cP(u^0))=\frac{\Big(p_z^0(\tau)+\mu \partial_{p_z} \cP^0\Big)^2}{2} - \lambda^2 \cos(z^0(\tau)+\mu \partial_{z} \cP^0 - \alpha) - \lambda^2\,,
\]
has no critical points if $\mu$ is small enough. Indeed, an easy computation shows that the derivative
\[
f'(0) = \frac{4 \lambda^2}{\cosh^2(\lambda \tau)} + \lambda^4 \sin^2(4 \arctan \rme^{\lambda \tau}+\pi)
\]
does not vanish in the region $\{p_x>0,p_y>0\}$ because $\lambda=(p_x^2+p_y^2)^{1/4}$.
\end{proof}

Before proving Proposition~\ref{prop:melnikov} we summarize some basic asymptotic
properties of the flows $\phi_t^0$ and $\phi_t^\ep$. We recall that the dynamics of the unperturbed
problem has a phase-shift
$
u_{\pm} = (x_{\pm},y_{\pm},0,0,0,0),
$
where $x_\pm$, and $y_\pm$ are given by Eq.~\eqref{eq:phase-shift}.
The trajectories on the invariant torus $\cT_{p_x,p_y}$ and the trajectories on the whiskers
$W^{\srm}(\cT_{p_x,p_y}) = W^{\urm}(\cT_{p_x,p_y})$ converge exponentially 
to each other, with rate $\lambda$ as $t\to\pm\infty$. More precisely
\begin{equation}\label{eq:prop:flows:0}
\begin{array}{ll}
|\phi_t^0(u^0)-\phi_t^0(u^*+u_+)|
\leq C_1 \rme^{-\lambda t} \,, &  t\geq
0\,, \\
|\phi_t^0(u^0)-\phi_t^0(u^*+u_-)|
\leq C_1 \rme^{-\lambda |t|}\,, &  t\leq
0\,, \\
\end{array}
\end{equation}
for some constant $C_1>0$. These estimates are obtained using
the explicit expressions computed in Section~\ref{ssec:geom:unper}, and they just reflect the normal
hyperbolicity of the NHIM. Analogous expressions hold for the perturbed system.
In this case, given $u^\srm \in W^\srm(\Lambda_\ep)$ and $u^\urm \in W^\urm(\Lambda_\ep)$, there exist
points on the NHIM,
$u^{\srm*}, u^{\srm}_+, u^{\urm*}, u^{\urm}_- \in \Lambda_\ep$,
that are $\ep$-close to their unperturbed counterparts.
These points satisfy
\begin{equation}\label{eq:prop:flows:ep}
\begin{array}{ll}
|\phi_t^\ep(u^\srm)-\phi_t^\ep(u^{\srm *}+u^\srm_+)|
\leq C_2 \rme^{-\lambda_\ep t}\,, &  t\geq
0\,, \\
|\phi_t^\ep(u^\urm)-\phi_t^\ep(u^{\urm *}+u^\urm_-)|
\leq C_2 \rme^{-\lambda_\ep |t|}\,, &  t\leq
0\,, \\
\end{array}
\end{equation}
for some constant $C_2>0$, where $\lambda_\ep=\lambda +\cO(\ep)$.
We also need to recall some estimates
that allow us to compare the perturbed and the unperturbed flows. The following
estimates, which hold for all $t\in\RR$, are standard and immediate to obtain
(for certain positive constants $C_3$, $C_4$, $C_5$ and $K$):
\begin{equation}\label{eq:est:NHIM}
\begin{split}
|\phi_t^\ep (u^\srm) - \phi_t^0 (u^0)| \leq {} & C_3 |u^\srm-u^0| \rme^{K \ep |t|} \leq C_5 \ep \rme^{K \ep |t|}\,, \\
|\phi_t^\ep (u^\urm) - \phi_t^0 (u^0)| \leq {} & C_3 |u^\urm-u^0| \rme^{K \ep |t|} \leq C_5 \ep \rme^{K \ep |t|}\,, \\
|\phi_t^\ep(u^{\srm *} + u^{\srm}_+) - \phi_t^0 (u^*+u_+)| \leq {} & C_4 |u^{\srm *}-u^*| \rme^{K \ep |t|} \leq C_5 \ep \rme^{K \ep |t|}\,, \\
|\phi_t^\ep(u^{\urm *} + u^{\urm}_-) - \phi_t^0 (u^*+u_-)| \leq {} & C_4 |u^{\urm *}-u^*| \rme^{K \ep |t|} \leq C_5 \ep \rme^{K \ep |t|}\,.
\end{split}
\end{equation}
These expressions state that for $\ep>0$ there may be unstable motions inside the NHIM and we
cannot have a global control on the dynamics for all time. 
Nevertheless, 
we have the bounds
\begin{equation}\label{eq:est:NHIM2}
C_5 \ep \rme^{K \ep |t|} \leq C_7 \ep^{\rho_1},
\qquad \mbox{for $|t|\leq C_6 \log (1/\ep)$},
\end{equation}
with $C_6>0$, $C_7>0$ and $0<\rho_1<1$, which is enough for our purposes.

\begin{proof}[Proof of Proposition~\ref{prop:melnikov}]
To monitor the evolution of the function $\cP$, given by Eq.~\eqref{eq:Ham:P}, along the perturbed flow, we use the formula
\[
\frac{d}{dt}(\cP(\phi_t^\ep(u)) = \{\cP,H_\ep\}(\phi_t^\ep(u)) = \ep \{\cP,h\}(\phi_t^\ep(u))\,,
\]
where we have used that $\{\cP,H_0\}=0$. Integrating this equation we obtain
\begin{equation}\label{eq:intP}
\cP(\phi_{t_2}^\ep (u)) = \cP(\phi^\ep_{t_1}(u)) + \ep \int_{t_1}^{t_2} \{\cP,h\}(\phi_\sigma^\ep(u)) d\sigma\,.
\end{equation}

Using~\eqref{eq:intP} with $t_2=0$, $t_1=\infty$ and $u=u^{\srm}$, we have
\begin{equation}\label{eq:Pu1}
\cP(u^{\srm})=\cP(\phi_{t \to \infty}^\ep (u^{\srm})) - \ep \int_0^\infty \{\cP,h\} ( \phi_{\sigma}^\ep(u^{\srm})) d\sigma\,,
\end{equation}
and using~\eqref{eq:intP} with $t_2=0$, $t_1=\infty$ and $u=u^{\srm *} + u^\srm_+$, we get
\[
\cP(u^{\srm *} + u^\srm_+)=\cP(\phi_{t \to \infty}^\ep (u^{\srm *} + u^\srm_+)) - \ep \int_0^\infty \{\cP,h\} ( \phi_{\sigma}^\ep(u^{\srm *} + u^\srm_+)) d\sigma\,.
\]
Subtracting these expressions we obtain
\begin{align}
 \cP(u^{\srm})-\cP(u^{\srm *}+u^{\srm}_+) = {} &  \cP(\phi_{t\to \infty}^\ep(u^{\srm}))-\cP(\phi_{t\to \infty}^\ep(u^{\srm *}+u^{\srm}_+))\nonumber\\ 
& - \ep \int_0^\infty 
\Big(\{\cP,h\} (\phi_\sigma^\ep(u^{\srm}))- \{\cP,h\} (\phi_\sigma^\ep(u^{\srm *}+u^{\srm}_+))\Big) d\sigma\,.\label{eqotra2}
\end{align}
Now, we observe that $\cP(u^{\srm *}+u^{\srm}_+)=\cP(u^*+u_++\cO(\ep))=\cO(\ep^2)$, 
since both $\cP$ and $\nabla \cP$ vanish on the unperturbed NHIM $\Lambda_0$.
Moreover, from the asymptotic properties~\eqref{eq:prop:flows:ep} it follows that
\[
| \cP(\phi_{t\to \infty}^\ep(u^{\srm}))-\cP(\phi_{t\to \infty}^\ep(u^{\srm *}+u^{\srm}_+))  | \leq C_8 |\phi_{t \to \infty}^\ep(u^{\srm})- \phi_{t\to \infty}^\ep(u^{\srm *}+u^{\srm}_+)| \to 0\,.
\]
To study the integral term, 
we recall that we cannot control the dynamics on the NHIM for all time,
so we consider
\begin{align*}
\int_{C_6 \log(1/\ep)}^\infty 
\Big(\{\cP,h\} (\phi_\sigma^\ep(u^{\srm})) & - \{\cP,h\} (\phi_\sigma^\ep(u^{\srm *}+u^{\srm}_+))\Big) d\sigma \\
\leq {} & \int_{C_6 \log(1/\ep)}^\infty C_9 | 
\phi_{\sigma}^\ep(u^{\srm}) - \phi_{\sigma}^\ep(u^{\srm *}+u^{\srm}_+)| d \sigma \\
\leq {} &  C_9 C_2 \int_{C_6 \log(1/\ep)}^\infty \mathrm{e}^{-\lambda_\ep \sigma} d \sigma = \frac{C_9 C_2}{\lambda_\ep} \mathrm{e}^{-\lambda_\ep C_6 \log(1/\ep)} = \cO(\ep^{\rho_2})\,,
\end{align*}
for certain constant $\rho_2>0$. Notice that we have used Eq.~\eqref{eq:prop:flows:ep} to derive the second inequality. Accordingly, these estimates and Eq.~\eqref{eqotra2} imply that
\[
\cP(u^{\srm}) = - \ep \int_0^{C_6 \log(1/\ep)} \Big(\{\cP,h\} (\phi_\sigma^\ep(u^{\srm}))- \{\cP,h\} (\phi_\sigma^\ep(u^{\srm *}+u^{\srm}_+))\Big) d\sigma + \cO(\ep^2) + \cO(\ep^{1+\rho_2})\,.
\]
Now, we can control the quantities $\phi_{\sigma}^\ep(u^{\srm})-\phi_{\sigma}^0(u^0)$ and $\phi_{\sigma}^\ep(u^{\srm *}+u^{\srm}_+)-\phi_\sigma^0 (u^*+u_+)$ using Eqs.~\eqref{eq:est:NHIM} and~\eqref{eq:est:NHIM2}, so we write
\begin{align*}
\cP(u^{\srm}) = {} &
- \ep \int_0^{C_6 \log(1/\ep)}
\Big(\{\cP,h\} (\phi_\sigma^0(u^{0})) - \{\cP,h\} (\phi_\sigma^0(u^{*}+u_+))\Big) d\sigma +I+\cO(\ep^2)+\cO(\ep^{1+\rho_2})\,,
\end{align*}
where
\begin{equation*}
I\leq \ep C_9 \int_0^{C_6 \log(1/\ep)}
\Big(|\phi_{\sigma}^\ep(u^{\srm})-\phi_{\sigma}^0(u^0)| +
|\phi_{\sigma}^\ep(u^{\srm *}+u^{\srm}_+)-\phi_\sigma^0 (u^*+u_+)|\Big) d\sigma\leq 
2C_9C_7C_6\ep^{1+\rho_1} \log(1/\ep) = \cO(\ep^{1+\rho_3})
\end{equation*}
for certain constant $0 <\rho_3 <1$. We conclude that
\[
\cP(u^\srm)=-\ep \int_0^\infty \Big(\{\cP,h\} (\phi_\sigma^0(u^{0})) - \{\cP,h\} (\phi_\sigma^0(u^{*}+u_+))\Big) d\sigma 
+ \cO(\ep^{1+\rho})\,,
\]
for some constant $\rho>0$. Here we have used the bound 
$$
\int_{C_6 \log(1/\ep)}^\infty 
\Big(\{\cP,h\} (\phi_\sigma^0(u^0)) - \{\cP,h\} (\phi_\sigma^0(u^{*}+u_+))\Big) d\sigma=\cO(\ep^{\rho_4})\,.
$$

Finally, obtaining a similar formula for $\cP(u^{\urm})$ and subtracting, we obtain
\begin{equation}\label{eq:int1}
\cP(u^{\urm})-\cP(u^{\srm}) = \ep \int_{-\infty}^{\infty} 
\Big(\{\cP,h\} (\phi_\sigma^0(u^{0})) - \{\cP,h\} (\phi_\sigma^0(u^{*}+u_\pm))\Big) d\sigma
+ \cO(\ep^{1+\rho})\,.
\end{equation}
Recalling that the unperturbed flow $\phi_{\sigma}^0$ satisfies Eq.~\eqref{eq:flow:H0:outer},
we can write
\[
\frac{\partial}{\partial \tau} (h(\phi_\sigma^0(u^0))) = 
\partial_x h^0 \, \dot{F}_1 (\tau+\sigma)
+\partial_y h^0 \, \dot{F}_2 (\tau+\sigma)
+\partial_z h^0 \, \dot{z}^0(\tau+\sigma)
+\partial_{p_z} h^0 \, \dot{p_z}(\tau+\sigma) \,,
\]
where we are using the notation $\partial_\xi h^0 := \frac{\partial h}{\partial \xi}(\phi_\sigma^0 (u^0))$. Now we observe that $\dot{F}_1(\tau) = - \sin(\alpha) - \sin (z^0(\tau))$,
and $\dot{F}_2(\tau) = - \cos(\alpha) - \cos (z^0(\tau))$,
so we obtain
\begin{align*}
\frac{\partial}{\partial \tau} (h(\phi_\sigma^0(u^0)))
= {} &
- [\sin(\alpha) + \sin (z^0(\tau+\sigma))] \partial_x h^0 - [ \cos(\alpha) + \cos (z^0(\tau+\sigma))] \partial_y h^0 \\
& + p_z^0(\tau+\sigma) \partial_z h^0  + (p_x \cos (z^0(\tau+\sigma))-p_y \sin(z^0(\tau+\sigma))) \partial_{p_z} h^0\,.
\end{align*}
Using the definition of $\cP$ in~\eqref{eq:Ham:P}, we end up with
\[
\frac{\partial}{\partial \tau} (h(\phi_\sigma^0(u^0)))
=  - \{\cP,h\}(\phi_\sigma^0 (u^0))\,.
\]
Hence, the expression~\eqref{eq:int1} is equivalent to
\[
\cP(u^{\urm})-\cP(u^{\srm}) = - \ep \frac{\partial}{\partial \tau} L(\tau,x,y,p_x,p_y)
+ \cO(\ep^{1+\rho})\,,
\]
where we 
have considered the expansion $h(u,\ep)=h(u,0)+\cO(\ep)$ and
used the definition of $L$ in Eq.~\eqref{eq:Poin:Mel}.
By Lemma~\ref{lem:heteroclinic}, homoclinic intersections are characterized by
the condition $\cP(u^\srm)=\cP(u^{\urm})$. Therefore, we conclude that the existence of
homoclinic intersections is given, at first order
perturbation theory, by the zeros of a directional derivative of the
Poincar\'e function $L(\tau,x,y,p_x,p_y)$, as we wanted to prove.
\end{proof}

Let us consider the Hamiltonian of the ABC system, written as
$
H_\ep= H_0 + \ep H_1 + \ep^2 H_2,
$
where $H_0$, $H_1$ and $H_2$ are given by Eqs.~\eqref{eq:H0:two}--~\eqref{eq:H2}. Then,
using expressions~\eqref{eq:flow:H0:inner} and~\eqref{eq:flow:H0:outer},
the Poincar\'e function $L$ has the form
\[
L(\tau,x,y,p_x,p_y)= M_1 \cos x + M_2 \cos y + M_3 \sin x + M_4 \sin y\,,
\]
where the coefficients $M_i\equiv M_i(\tau,p_x,p_y)$ are given by the integrals
\begin{align}
M_1 := {} & \hat B \int_{-\infty}^\infty \bigg( 
(p_y - \cos z^*) \sin (x_{\pm}+\omega_1 \sigma) 
-(p_y-\cos z^0) \sin (F_1 + \omega_1 \sigma) 
- p_z^0 \cos(F_1 + \omega_1 \sigma) 
\bigg) d\sigma\,, \label{eq:M1:simple}\\
M_2 := {} & \hat C \int_{-\infty}^\infty \bigg( 
(p_x - \sin z^*) \cos (y_{\pm}+\omega_2 \sigma) 
-(p_x-\sin z^0) \cos (F_2 + \omega_2 \sigma) 
- p_z^0 \sin(F_2 + \omega_2 \sigma) 
\bigg) d\sigma\,, \label{eq:M2:simple}\\
M_3 := {} & \hat B \int_{-\infty}^\infty \bigg( 
(p_y - \cos z^*) \cos (x_{\pm}+\omega_1 \sigma) 
-(p_y-\cos z^0) \cos (F_1 + \omega_1 \sigma) 
+ p_z^0 \sin(F_1 + \omega_1 \sigma) 
\bigg) d\sigma\,, \label{eq:M3:simple}\\
M_4 := {} &  \hat C \int_{-\infty}^\infty \bigg( 
(p_x-\sin z^0) \sin (F_2 + \omega_2 \sigma) 
- (p_x - \sin z^*) \sin (y_{\pm}+\omega_2 \sigma) 
- p_z^0 \cos(F_2 + \omega_2 \sigma) 
\bigg) d\sigma\,. \label{eq:M4:simple}
\end{align}
Here $F_1=F_1(\tau+\sigma)$ and $F_2=F_2(\tau+\sigma)$ are given by Eq.~\eqref{eq:F1:F2:ABC},
and $z^0=z^0(\tau+\sigma)$ and $p_z^0=p_z^0(\tau+\sigma)$
are given by Eq.~\eqref{eqsep}.

\subsection{The Scattering map}
\label{ssec:scattering}

It is convenient to introduce the notation
\begin{equation}\label{redpoin}
\cL(x-\omega_1 \tau, y-\omega_2 \tau,p_x,p_y):=L(0,x-\omega_1 \tau, y-\omega_2 \tau, p_x, p_y)\,.
\end{equation}
Since the properties of the unperturbed flow imply that 
$$
L(0,x-\omega_1 \tau, y-\omega_2 \tau, p_x, p_y)=L(\tau,x,y,p_x,p_y)\,,
$$
we can consider the critical points of the function
\begin{equation}\label{eq:tau:poincaremelnikoc}
\tau \longmapsto \cL(x-\omega_1 \tau, y-\omega_2 \tau,p_x,p_y)
\end{equation}
in order to study the homoclinic intersections.

Then, we introduce the domain $\cD \subset \TT^2 \times \cI\subset \TT^2\times \RR^2$ in hypothesis $\mathbf{A}_2$ of Theorem~\ref{teo:diffusion:ABC}, such that
for each $(x,y,p_x,p_y)$ in $\cD$, there exists a unique critical point
$\tau^*=\tau^*(x,y,p_x,p_y)$ of the map~\eqref{eq:tau:poincaremelnikoc}
defining a 
smooth function
on $\cD$.
This implies that the points
\[
(x,y,z^0(\tau^*),p_x,p_y,p_z^0(\tau^*)) + \cO(\ep) \in W^\srm(\Lambda_\ep) \pitchfork W^\urm(\Lambda_\ep) 
\]
define a manifold $\Gamma_\ep$, called \emph{homoclinic manifold}. The scattering map
associated to $\Gamma_\ep$ is defined in a domain $\cD_{\ep,b} \subset \cD$ in the following
way (see~\cite{DLS06,DLS08}):
\begin{equation}\label{eq:scat:map}
\begin{array}{rcl}
s_\ep : \cD_{\ep,{\rm b}} \subset \cD & \longrightarrow & \cD_{\ep,{\rm f}} \subset \TT^2 \times \RR^2\,,\\
u_{\rm b} & \longmapsto & u_{\rm f}\,,
\end{array}
\end{equation}
with $u_{\rm f}=s_{\ep}(u_{\rm b})$ if and only if there exists $u \in \Gamma_\ep$ such that
\begin{align*}
|\phi_t^\ep(u)-\phi_t^\ep(P_\ep(u_{\rm f}))| & \longrightarrow 0, \quad t \rightarrow \infty\,, \\
|\phi_t^\ep(u)-\phi_t^\ep(P_\ep(u_{\rm b}))| & \longrightarrow 0, \quad t \rightarrow -\infty\,,
\end{align*}
where $P_\ep$ is the parameterization of the perturbed NHIM $\Gamma_\ep$ introduced in Section~\ref{ssec:inner:approx}. Since the parameterizing variables $(x,y,p_x,p_y)$ and the phase space variables coincide up to order $\ep^2$, we can safely assume that they are the same. 
The sets $\cD_{\ep,{\rm b}}$ and $\cD_{\ep,{\rm f}}$ are defined as:
\[
\cD_{\ep,{\rm b}} := \bigcup_{u \in \Gamma_\ep} \{u_{\rm b}\},
\qquad
\cD_{\ep,{\rm f}} := \bigcup_{u \in \Gamma_\ep} \{u_{\rm f}\}.
\]
The scattering map relates the past asymptotic trajectory of any orbit in the homoclinic
manifold to its future asymptotic behavior.

The scattering map~\eqref{eq:scat:map} is exact symplectic (see~\cite{DLS08})
and it is given by the time-1 flow of the Hamiltonian function
\[
S_\ep = S_0 + \ep S_1 + \cO(\ep^2)
\]
where $S_0$ corresponds to the unperturbed outer dynamics,
and $S_1$ is given by the Poincar\'e function~\eqref{redpoin} evaluated at $\tau=\tau^*$. Notice that
the unperturbed scattering map for the ABC system satisfies
\[
u_{\rm f} = u_{\rm b} + u_{+} - u_{-}, \qquad 
u_{\pm} = (x_{\pm},y_{\pm},0,0,0,0)\,,
\]
where $x_\pm$, and $y_\pm$ are given by Eq.~\eqref{eq:phase-shift}.
Hence, we obtain the following expressions for $S_0$ and $S_1$:
\[
S_0=-8 (p_x^2 + p_y^2)^{1/4},
\qquad
S_1(x,y,p_x,p_y)=\cL(x-\omega_1 \tau^*, y-\omega_2 \tau^*,p_x,p_y)\,,
\]
where $\tau^*=\tau^*(x,y,p_x,p_y)$ is the critical point
of the function~\eqref{eq:tau:poincaremelnikoc} that
has the following expression for the ABC system:
\[
\tau \mapsto M_1^0 \cos(x-\omega_1 \tau) + M_2^0 \cos(y-\omega_2 \tau) + M_3^0 \sin(x-\omega_1 \tau) + M_4^0 \sin (y-\omega_2\tau)\,,
\]
where $M_i^0:= M_i(0,p_x,p_y)$ are obtained evaluating the integrals~\eqref{eq:M1:simple},
\eqref{eq:M2:simple}, \eqref{eq:M3:simple}, and \eqref{eq:M4:simple}. 
Finally, we discuss some conditions that allows us to justify that there
exists a domain $\cD$ where the above construction is well posed for the ABC system.
We fix a value of $(p_x,p_y)$ and notice that the function
\[
(x,y) \mapsto \cL(x,y) = M_1^0 \cos(x) + M_2^0 \cos(y) + M_3^0 \sin(x) + M_4^0 \sin (y),
\]
has four critical points $(x_c,y_c)$ given by
\begin{equation}\label{eq:crit:cL}
x_c = \arctan \frac{M_3^0}{M_1^0}, \qquad y_c = \arctan \frac{M_4^0}{M_2^0}\,.
\end{equation}
It is easy to check that these critical points are nondegenerate provided that $M_1^0$ and $M_3^0$ do not vanish simultaneously, and the same for $M_2^0$ and $M_4^0$. This implies, in particular, that $\hat B\neq0$ and $\hat C\neq 0$, as required in the statement of Theorem~\ref{teo:diffusion:ABC}.
Hence, we observe that we are in the same
situation considered in~\cite{DH11}, where the existence of $\tau^*$ was
justified in detail using the tangential intersection of straight lines in the direction $(\omega_1,\omega_2)$
with the regular level curves of $\cL(x,y)$, which are periodic curves which fill out
a region bounded by the level curves containing the saddle points.

\section{Combination of inner and outer dynamics}
\label{sec:chains}

In this Section we conclude the proof of Theorem~\ref{teo:diffusion:ABC}
and we give explicit formulas for the condition~\eqref{eq:cond:trans:teo}.
To this end, we combine the inner and outer dynamics.
In Proposition~\ref{prop:level:sets} we showed that the invariant tori (both primary
and secondary) of the ABC system are given by the level sets
of a couple of functions. This couple defines an $\RR^2$-valued map that will be denoted as $F_\ep$ all along this section. The scattering map described in Section~\ref{ssec:scattering}
transports the level sets of $F_\ep$ onto the level sets of $F_\ep \circ s_\ep$. Then, following~\cite{DLS00,DLS06}, it turns out that (c.f. Lemma 10.4 in~\cite{DLS06})
given two manifolds $\Sigma_1,\Sigma_2 \subset \Lambda_\ep$
that are invariant under the inner dynamics, if $\Sigma_1$ intersects
transversally $s_\ep(\Sigma_2)$ in $\Lambda_\ep$, then $W^{\srm}_{\Sigma_1} \pitchfork W^{\srm}_{\Sigma_2}$.
This is the main ingredient to create heteroclinic intersections between the KAM tori in $\Lambda_\ep$.

To characterize the action of the scattering map on the level sets of a given function, we
follow the computations in~\cite{DH09} (which are also used in~\cite{DLS13,DH11}). Given a function
$F=F_0 + \ep F_1 + \ep^2 F_2 + \ldots$ we can approximate $F \circ s_\ep$ as
\begin{align}
F \circ s_\ep = {} & F + \{F,S_\ep\} + \cO(\ep^2) = F + \{ F_0 + \ep F_1, S_0 + \ep S_1\} +\cO(\ep^2) \nonumber \\
              = {} & F + \{ F_0, S_0\} + \ep (\{F_1, S_0\} + \{ F_0,S_1\} ) +\cO(\ep^2)\,. \label{eq:F:S}
\end{align}
It is worth mentioning that this expression does not correspond with the expression
obtained in~\cite{DLS13,DH09,DH11}, due to the presence of a phase-shift in the unperturbed problem.
We observe that $F_0$ and $S_0$ depend
only on the momenta, so we have $\{F_0,S_0\}=0$.

Let us consider a function $F : \TT^2 \times \RR^2 \rightarrow \RR^2$ that defines the level sets
\begin{align}
F_1 = F_{1,0} + \ep F_{1,1} + \cO(\ep^2) = {} & e_1\,, \label{eq:trans1}\\
F_2 = F_{2,0} + \ep F_{2,1} + \cO(\ep^2) = {} & e_2\,, \label{eq:trans2}
\end{align}
and the transformation of these level sets by means of the scattering map
\begin{align}
F_1 \circ s_\ep = F_{1,0} + \ep F_{1,1} + \ep (\{F_{1,1}, S_0\} + \{ F_{1,0},S_1\} ) + \cO(\ep^2) = {} & e_1'\,, \label{eq:trans3}\\
F_2 \circ s_\ep = F_{2,0} + \ep F_{2,1} + \ep (\{F_{2,1}, S_0\} + \{ F_{2,0},S_1\} ) + \cO(\ep^2) = {} & e_2'\,, \label{eq:trans4}
\end{align}
for certain $e_1', e_2'$. We will use Eqs.~\eqref{eq:trans1}--~\eqref{eq:trans4}
to determine transversal intersections between these level sets. Indeed, if we subtract these expressions, we
have
\begin{align}
\{F_{1,1}, S_0\} + \{ F_{1,0},S_1\} + \cO(\ep) = {} & \frac{e_1'-e_1}{\ep}\,, \label{eq:trans5}\\
\{F_{2,1}, S_0\} + \{ F_{2,0},S_1\} + \cO(\ep) = {} & \frac{e_2'-e_2}{\ep}\,. \label{eq:trans6}
\end{align}
Then, if we use Eqs.~\eqref{eq:trans1} and~\eqref{eq:trans2} to write $p_x=p_x(x,y,e_1,e_2)$ and
$p_y=p_y(x,y,e_1,e_2)$, and we introduce these expressions into Eqs.~\eqref{eq:trans5} and~\eqref{eq:trans6}, it turns out that we will have intersection as long as
$\tfrac{e_1'-e_1}{\ep}$ and $\tfrac{e_1'-e_1}{\ep}$ are small enough, close to
the non-degenerate solutions of
\begin{align}
\{F_{1,1}, S_0\} + \{ F_{1,0},S_1\} = {} & 0\,, \label{eq:trans7}\\
\{F_{2,1}, S_0\} + \{ F_{2,0},S_1\} = {} & 0\,. \label{eq:trans8}
\end{align}
The non-degeneracy condition, which implies that the intersection is transversal, reads as
\begin{equation}\label{eq:det}
\det
\begin{pmatrix}
\vphantom{\bigg(} \frac{\partial}{\partial x} \Big(\{F_{1,1}, S_0\} + \{ F_{1,0},S_1\} \Big) & \frac{\partial}{\partial y} \Big(\{F_{1,1}, S_0\} + \{ F_{1,0},S_1\} \Big) \\
\frac{\partial}{\partial x} \Big(\{F_{2,1}, S_0\} + \{ F_{2,0},S_1\} \Big) & \frac{\partial}{\partial y} \Big(\{F_{2,1}, S_0\} + \{ F_{2,0},S_1\} \Big) 
\end{pmatrix}
\neq 0
\end{equation}
for each point at the intersection of the level sets of $F$. Let us remark that the condition~\eqref{eq:det}
is evaluated by fixing $p_x$ and $p_y$ by means of $F_{1,0}=e_1$ and $F_{2,0}=e_2$, where
we have neglected the $\cO(\ep)$ terms.

\begin{remark}\label{rem:symm}
It is easy to check that, in the non-resonant case, the matrix in Eq.~\eqref{eq:det} is symmetric.
This is a consequence of the geometric structure of the problem, since the functions
$F_{1,1}$ and $F_{2,1}$ are obtained by means of the change of variables~\eqref{eq:change:g}.
\end{remark}

Finally, let us express the condition~\eqref{eq:det} in a explicit way
for the case of the ABC system. We consider separately the non-resonant and the resonant zones:

\begin{itemize}
\item \textbf{Non-resonant region:} From Proposition~\ref{prop:level:sets} it follows that
we have to consider the level sets~\eqref{eq:tori:non-resonant}. In order to check the condition~\eqref{eq:det}
we introduce 
$F_1=F_{1,0} + \ep F_{1,1} + \cO(\ep^2)$
and 
$F_2=F_{2,0} + \ep F_{2,1} + \cO(\ep^2)$,
where
\begin{align*}
F_{1,0} := {} & p_x\,, \\
F_{1,1} := {} & - \hat B \tfrac{\omega_2}{\omega_1} \sin x\,, \\
F_{2,0} := {} & p_y\,, \\
F_{2,1} := {} & - \hat C \tfrac{\omega_1}{\omega_2} \cos y\,.
\end{align*}
A direct computation shows that
\begin{align*}
\{F_{1,1},S_0\} = {} & 4 \hat B \tfrac{\omega_2}{\omega_1} p_x (p_x^2 + p_y^2)^{-3/4} \cos x\,, \\
\{F_{1,0},S_1\} = {} & 
-M_3^0 \left( 1-\omega_1 \tfrac{\partial \tau^*}{\partial x} \right) \cos(x-\omega_1 \tau^*)
+M_4^0 \omega_2 \tfrac{\partial \tau^*}{\partial x} \cos(y-\omega_2 \tau^*) \\
& +M_1^0 \left( 1-\omega_1 \tfrac{\partial \tau^*}{\partial x} \right) \sin(x-\omega_1 \tau^*)
-M_2^0 \omega_2 \tfrac{\partial \tau^*}{\partial x} \sin(y-\omega_2 \tau^*)\,,\\
\{F_{2,1},S_0\} = {} & - 4 \hat C \tfrac{\omega_1}{\omega_2} p_y (p_x^2 + p_y^2)^{-3/4} \sin y\,, \\
\{F_{2,0},S_1\} = {} & 
M_3^0 \omega_1 \tfrac{\partial \tau^*}{\partial y} \cos(x-\omega_1 \tau^*)
-M_4^0 \left( 1-\omega_2 \tfrac{\partial \tau^*}{\partial y} \right) \cos(y-\omega_2 \tau^*) \\
& -M_1^0 \omega_1 \tfrac{\partial \tau^*}{\partial y} \sin(x-\omega_1 \tau^*)
+M_2^0 \left(1- \omega_2 \tfrac{\partial \tau^*}{\partial y} \right) \sin(y-\omega_2 \tau^*)\,,
\end{align*}
and a straightforward but cumbersome computation allows us to compute the $2\times 2$ matrix in Eq.~\eqref{eq:det}, which reads as
\begin{equation*}
\left(
  \begin{array}{cc}
    \Delta_1 & \Delta_2 \\
    \Delta_2 & \Delta_3 \\
  \end{array}
\right)\,,
\end{equation*}
where the coefficients have the expressions
{\allowdisplaybreaks
\begin{align}
\Delta_1 := {} & \tfrac{\partial}{\partial x} (\{F_{1,1}, S_0\} + \{ F_{1,0},S_1\}) \label{eq:Delta:1} \\
= {} & - 4 \hat B \tfrac{\omega_2}{\omega_1} p_x (p_x^2 + p_y^2)^{-3/4} \sin x + \left[ M_3^0 \omega_1 \tau^*_{xx}+M_1^0 (1-\omega_1 \tau^*_x)^2 \right] \cos(x-\omega_1 \tau^*)  \nonumber \\
& + \left[ M_4^0 \omega_2 \tau^*_{xx}+M_2^0 (\omega_2 \tau^*_x)^2 \right] \cos(y-\omega_2 \tau^*)+ \left[ -M_1^0 \omega_1 \tau^*_{xx} + M_3^0 (1-\omega_1 \tau^*_x)^2 \right] \sin(x-\omega_1 \tau^*)  \nonumber \\
& + \left[ -M_2^0 \omega_2 \tau^*_{xx} + M_4^0 (\omega_2 \tau^*_{x})^2 \right] \sin(y-\omega_2 \tau^*)\,, \nonumber \\
\Delta_2 := {} & \tfrac{\partial}{\partial y} (\{F_{1,1}, S_0\} + \{ F_{1,0},S_1\}) = \tfrac{\partial}{\partial x} (\{F_{2,1}, S_0\} + \{ F_{2,0},S_1\}) \label{eq:Delta:2}\\
= {} & \left[M_3^0 \omega_1 \tau^*_{xy}-M_1^0 \omega_1 \tau^*_{y} (1-\omega_1 \tau^*_x) \right] \cos(x-\omega_1 \tau^*)  \nonumber \\
& + \left[M_4^0 \omega_2 \tau^*_{xy}-M_2^0 \omega_2 \tau^*_{x} (1-\omega_2 \tau^*_y) \right] \cos(y-\omega_2 \tau^*)  \nonumber \\
& + \left[-M_1^0 \omega_1 \tau^*_{xy}-M_3^0 \omega_1 \tau^*_{y} (1-\omega_1 \tau^*_x) \right] \sin(x-\omega_1 \tau^*)  \nonumber \\
& + \left[-M_2^0 \omega_2 \tau^*_{xy}-M_4^0 \omega_2 \tau^*_{x} (1-\omega_2 \tau^*_y) \right] \sin(y-\omega_2 \tau^*)\,,
\nonumber \\
\Delta_3 := {} &  \tfrac{\partial}{\partial y} (\{F_{2,1}, S_0\} + \{ F_{2,0},S_1\}) 
\label{eq:Delta:3} \\
= {} & -4 \hat C \tfrac{\omega_1}{\omega_2} p_y (p_x^2 + p_y^2)^{-3/4} \cos y + \left[M_3^0 \omega_1 \tau^*_{yy}+M_1^0 (\omega_1 \tau^*_{y})^2 \right] \cos(x-\omega_1 \tau^*)  \nonumber \\
& + \left[M_4^0 \omega_2 \tau^*_{yy}+M_2^0 (1-\omega_2 \tau^*_y)^2 \right] \cos(y-\omega_2 \tau^*) + \left[-M_1^0 \omega_1 \tau^*_{yy}+M_3^0 (\omega_1 \tau^*_{y})^2 \right] \sin(x-\omega_1 \tau^*)  \nonumber \\
& + \left[-M_2^0 \omega_2 \tau^*_{yy}+M_4^0 (1-\omega_2 \tau^*_y)^2 \right] \sin(y-\omega_2 \tau^*)\,.
\nonumber
\end{align}}

\vspace{-0.5cm}
\noindent Here the subscripts in $\tau^*$ denote, as usual, partial differentiation.
Then, the transversality condition in the non-resonant
region, using the functions $\Delta_i$, reads as
\begin{equation}\label{eq:cond:trans1}
\Delta_1 \Delta_3-\Delta_2^2 \neq 0.
\end{equation}

\item \textbf{Resonant region:}
From Proposition~\ref{prop:level:sets} it follows that we have to consider the level sets~\eqref{eq:tori:resonant}.
In order to check the condition in Eq.~\eqref{eq:det}
we introduce
$F_1=F_{1,0} + \ep F_{1,1} + \cO(\ep^2)$,
and 
$F_2=F_{2,0} + \ep F_{2,1} + \cO(\ep^2)$,
where
\begin{align*}
F_{1,0} := {} & p_x+p_y\,, \vphantom{\tfrac{\omega_1}{\omega_1}} \\
F_{1,1} := {} &  - \hat B \tfrac{\omega_2}{\omega_1} \sin x - \hat C \tfrac{\omega_1}{\omega_2} \cos y\,, \\
F_{2,0} := {} & \tfrac{e_1^2}{4}+\tfrac{e_1}{\sqrt{2}} + (1+\sqrt{2} e_1^{-1}) \left(\tfrac{e_1}{2}-p_y \right)^2\,, \\
F_{2,1} := {} & 2 \hat C \tfrac{\omega_1}{\omega_2} (1+\sqrt{2} e_1^{-1}) \left(\tfrac{e_1}{2}-p_y \right) \cos y\,.
\end{align*}
Then, the matrix in Eq.~\eqref{eq:det} has the form
\begin{equation*}
\left(
  \begin{array}{cc}
    \hat\Delta_1 & \hat\Delta_2 \\
    \hat\Delta_3 & \hat\Delta_4 \\
  \end{array}
\right)\,,
\end{equation*}
with 
{\allowdisplaybreaks
\begin{align}
\hat \Delta_1 := {} & \tfrac{\partial}{\partial x} (\{F_{1,1}, S_0\} + \{ F_{1,0},S_1\})  
= -4 \hat B \tfrac{\omega_2}{\omega_1} p_x (p_x^2 + p_y^2)^{-3/4} \sin x  
\label{eq:Delta:4} \\
& + \left[ M_3^0 \omega_1 ( \tau^*_{xx} +\tau^*_{xy})
+ M_1^0 (1-\omega_1 \tau^*_x)(1-\omega_1(\tau^*_x + \tau^*_y))
\right] \cos(x-\omega_1 \tau^*)  \nonumber \\
& + \left[ M_4^0 \omega_2 (\tau^*_{xx}+\tau^*_{xy})
- M_2^0 \omega_2 \tau^*_x (1-\omega_2(\tau^*_x + \tau^*_y)) \right]\cos(y-\omega_2 \tau^*)  \nonumber \\
& + \left[ -M_1 \omega_1(\tau^*_{xx} + \tau^*_{xy}) + M_3^0 (1-\omega_1 \tau^*_x)
(1-\omega_1(\tau^*_x + \tau^*_y)) \right] \sin (x-\omega_1 \tau^*)  \nonumber \\
& + \left[
- M_2^0 \omega_2 (\tau^*_{xx}+\tau^*_{xy})
- M_4^0 \omega_2 \tau^*_x (1- \omega_2(\tau^*_x + \tau^*_y)) 
\right]\sin (y-\omega_2 \tau^*)\,, 
\nonumber \\
\hat \Delta_2 := {} & \tfrac{\partial}{\partial y} (\{F_{1,1}, S_0\} + \{ F_{1,0},S_1\})
= -4 \hat C \tfrac{\omega_1}{\omega_2} p_y (p_x^2 + p_y^2)^{-3/4} \cos y  
\label{eq:Delta:5}  \\
& + \left[ M_3^0 \omega_1 (\tau_{xy}^*+\tau^*_{yy})-M_1^0 \omega_1 \tau^*_y (1-\omega_1(\tau^*_x+\tau^*_y)) \right] \cos(x-\omega_1 \tau^*)  \nonumber \\
& + \left[ M_4^0 \omega_2 (\tau^*_{xy}+\tau^*_{yy}) + M_2^0 (1-\omega_2 \tau^*_y) (1-\omega_2(\tau^*_x+\tau^*_y)) \right]\cos(y-\omega_2 \tau^*)  \nonumber \\
& + \left[ -M^0_1 \omega_1 (\tau^*_{xy}+\tau^*_{yy}) - M_3^0\omega_1 \tau^*_y (1-\omega_1(\tau^*_x+\tau^+_y))\right] \sin (x-\omega_1 \tau^*) \nonumber \\
& + \left[ -M^0_2 \omega_2 (\tau^*_{xy}+\tau^*_{yy}) + M_4^0 (1-\omega_2 \tau^*_y)(1-\omega_2(\tau^*_x+\tau^*_y))\right] \sin(y-\omega_2 \tau^*)\,, 
\nonumber \\
\hat \Delta_3 := {} & \tfrac{\partial}{\partial x} (\{F_{2,1}, S_0\} + \{ F_{2,0},S_1\}) 
= \gamma \left[ -M^0_3 \omega_1 \tau^*_{xy} + M_1^0 \omega_1 \tau^*_y (1-\omega_1 \tau^*_x) \right] \cos(x-\omega_1 \tau^*)
\label{eq:Delta:6} \\
{} & +\gamma \left[ - M_4^0 \omega_2 \tau^*_{xy} + M^0_2 \omega_2 \tau^*_x (1-\omega_2 \tau^*_y) \right] \cos(y-\omega_2 \tau^*)  \nonumber \\
{} & +\gamma \left[ M_1^0 \omega_1 \tau^*_{xy} + M_3^0 \omega_1 \tau^*_y (1-\omega_1 \tau^*_x) \right]  \sin (x-\omega_1 \tau^*)  \nonumber \\
{} & +\gamma \left[ M^0_2 \omega_2 \tau^*_{xy} + M^0_4  \omega_2 \tau^*_x (1-\omega_2 \tau^*_y) \right] \sin (y-\omega_2 \tau^*)\,, 
\nonumber \\
\hat \Delta_4 := {} & \tfrac{\partial}{\partial y} (\{F_{2,1}, S_0\} + \{ F_{2,0},S_1\}) 
= 8 \hat C p_y (p_x^2 + p_y^2)^{-3/4} \tfrac{\omega_1}{\omega_2} (1+\tfrac{\sqrt{2}}{p_x+p_y}) \left(\tfrac{p_x-p_y}{2} \right) \cos y  
\label{eq:Delta:7} \\
& + \gamma \left[ - M^0_3 \omega_1 \tau^*_{yy} - M^0_1 (\omega_1 \tau^*_y)^2 \right] \cos(x-\omega_1 \tau^*)  \nonumber \\
& + \gamma \left[ - M^0_4 \omega_2 \tau^*_{yy} - M^0_2 (1-\omega_2 \tau^*_y)^2 \right] \cos(y-\omega_2 \tau^*)  \nonumber \\
& + \gamma \left[  M^0_1 \omega_1 \tau^*_{yy} - M^0_3 (\omega_1 \tau^*_y)^2 \right] \sin(x-\omega_1 \tau^*) \nonumber  \\
& + \gamma \left[  M^0_2 \omega_2 \tau^*_{yy} - M^0_4 (1-\omega_2 \tau^*_y)^2 \right] \sin(y-\omega_2 \tau^*)\,.
\nonumber
\end{align}}

\vspace{-0.5cm}
\noindent Here we are using the notation
$
\gamma:=2(1+\sqrt{2}) \left(\tfrac{p_x-p_y}{2}\right).
$
Then, the transversality condition in the resonant region takes
the form
\begin{equation}\label{eq:cond:trans2}
\hat \Delta_1 \hat \Delta_4-\hat \Delta_2 \hat \Delta_3 \neq 0\,.
\end{equation}
\end{itemize}

Putting together the information gathered on the inner and the outer
dynamics, we can construct chains of invariant tori giving rise to large motions
in the action space. Assume that $\cD$ is the domain introduced in Section~\ref{ssec:scattering}. It is obvious that we can safely assume, by shrinking $\cD$ if necessary, that the domain $\cD_{\ep,\rm b}$ of the scattering map~\eqref{eq:scat:map} coincides with $\cD$. Then, we assume that we can choose a constant $L$ such that
such that for every $(x,y,p_x,p_y)\in \cD$ we have
\[
\left\{
\begin{array}{ll}
\Delta_1 \Delta_3 - \Delta_2^2 \neq 0, & \mbox{if $|p_x-p_y|\geq L$}\,, \\
\hat \Delta_1 \hat \Delta_4 -\hat \Delta_2 \hat \Delta_3 \neq 0, & \mbox{if $|p_x-p_y|\leq L$}\,.
\end{array}
\right.
\]
This condition is precisely Eq.~\eqref{eq:cond:trans:teo} in Theorem~\ref{teo:diffusion:ABC}. It then follows that we can find a sequence
$\{\cT_i\}_{i=0}^\infty$ of tori which are at a
distance $\cO(\ep)$ from each other and that satisfy $s_\ep(\cT_i) \pitchfork \cT_{i+1}$.
By applying Lemma 10.4 in~\cite{DLS06}, it turns out that these tori satisfy $W^{\urm}_{\cT_i} \pitchfork W^{\srm}_{\cT_{i+1}}$, that is, they form a transition chain. The claim of Theorem~\ref{teo:diffusion:ABC} (Arnold diffusion) then follows from the general theory presented in~\cite{DLS06,DLS13}.

\section{Rigorous verification of the hypotheses of the main theorem}\label{sec:explicit}

In this section we illustrate the effective verification of the hypotheses
of Theorem~\ref{teo:diffusion:ABC}, thus obtaining Arnold diffusion in the ABC
Hamiltonian system~\eqref{eq:scaled:ham}.
We rigorously evaluate the involved functions and obtain rigorous bounds
for the critical points with the help of the computer. 
Our approach is as simple as possible, in the sense that we
do not pretend to present a fast and efficient methodology to
study large regions of the phase-space systematically. Our interest here is to convince
the general reader that the hypotheses of Theorem~\ref{teo:diffusion:ABC}
can be rigorously checked with the help of a computer.

In rigorous computations, real numbers are substituted by intervals whose extrema
are computer representable real numbers. That is, when implementing interval operations in a computer, the result of an operation with intervals is an interval that includes the result. The reader can consult the recent introductory book~\cite{Tucker11} on rigorous computations. All the computations presented in this
section have been performed using FILIB~\cite{HofschusterK97}
that uses double precision arithmetics.

Rigorous bounds of the Melnikov coefficients are obtained in Section~\ref{ssec:cap:Mi}.
In Section~\ref{ssec:cap:tau} we control the critical point $\tau^*$ and its
derivatives with respect to the angles $(x,y)$. Finally, in Section~\ref{ssec:cap:delta}
we present a direct application of the previous ideas and we describe the
implementation details giving rise to Corollary~\ref{cor:diffusion:ABC:informal}.

\subsection{On the evaluation of the Melnikov coefficients}
\label{ssec:cap:Mi}

Given certain values of $p_x$ and $p_y$ (represented using interval arithmetics),
we are interested in the rigorous evaluation of the Melnikov coefficients $M_1$, $M_2$, $M_3$
and $M_4$ in Eqs.~\eqref{eq:M1:simple}--~\eqref{eq:M4:simple}. Let us recall that we are particularly
interested in the values $M_i^0=M_i(0,p_x,p_y)$ in
order to check the hypotheses of Theorem~\ref{teo:diffusion:ABC}.
If we denote by $f_i(\sigma)$ the function that we have to integrate to
evaluate $M_i^0$, and we introduce the notation $f_{i,+}(\sigma)=f_i(\sigma)$ for $\sigma>0$ and
$f_{i,-}(\sigma)=f_i(\sigma)$ for $\sigma<0$, we can write the expressions for
$M_i^0$ as follows
\[
M_i^0 = \int_{-\infty}^{\infty} f_i(\sigma) d\sigma = 
\int_{-\infty}^{-a} f_{i,-}(\sigma)d\sigma
+\int_{-a}^0 f_{i,-}(\sigma)d\sigma
+\int_0^{a} f_{i,+}(\sigma)d\sigma
+\int_{a}^{\infty} f_{i,+}(\sigma)d\sigma\,,
\]
where $a>0$ is a constant that will be fixed later.
The integrals at infinity (called \emph{tails} from now on) will be bounded
using the asymptotic properties discussed in
Section~\ref{ssec:geom:unper}. Of course, one can obtain general formulas for the tails in terms of a uniform control on the Hamiltonian and the Lyapunov exponent.
However, in this case, we present specific formulas for the ABC
system giving rise to sharper estimates of the tails. 
This allows us to keep the modulus of the tails under a prefixed tolerance
using a small value of $a$.

\begin{lemma}\label{lem:control:M}
The following bounds hold for the ABC system:
\[
\left | M_i^0 -\int_{-a}^0 f_{i,-}(\sigma)d\sigma - 
\int_0^{a} f_{i,+}(\sigma)d\sigma \right| \leq \Sigma_i\,,
\]
where $\Sigma_i$ are given by
\begin{align*}
\Sigma_1 = \Sigma_3 :={} & \hat B 
\frac{4}{\lambda^2} 
|p_y-\cos z^*|
\left|
\log(1+\rme^{-2 \lambda a}) \sin z^*
- 2 \bigg(\arctan \Big(\frac{\rme^{-\lambda a}-1}{\rme^{-\lambda a}+1}\Big) + \frac{\pi}{4} \bigg) \cos z^*
\right| + \hat B \Sigma_0\,, \\
\Sigma_2 = \Sigma_4 :={} & \hat C 
\frac{4}{\lambda^2} 
|p_x-\sin z^*|
\left|
\log(1+\rme^{-2 \lambda a}) \cos z^*
+ 2 \bigg(\arctan \Big(\frac{\rme^{-\lambda a}-1}{\rme^{-\lambda a}+1}\Big) + \frac{\pi}{4} \bigg) \sin z^*
\right| + \hat C \Sigma_0\,,
\end{align*}
where $\Sigma_0:= 8 \tfrac{\rme^{-\lambda a}}{\lambda} + \tfrac{8}{3} \tfrac{\rme^{-3 \lambda a}}{\lambda} +  4 \hat B |\arctan(\sinh(\lambda a))-\pi|$.
\end{lemma}

\begin{proof}
We will only consider the case of $M_1^0$ because the control of the tails of $M_2^0$, $M_3^0$ and $M_4^0$ is completely analogous. We first split the function $f_{1,+}$ into three terms as
\[
f_{1,+}(\sigma)= f_{1,+}^{(1)}+f_{1,+}^{(2)}+f_{1,+}^{(3)}\,.
\]
In this splitting each term is given by
\begin{align*}
f_{1,+}^{(1)} := {} & \hat B (p_y-\cos z^*) (\sin(x_+ + \omega_1 \sigma)-\sin(F_1(\sigma)+\omega_1 \sigma))\,,\\
f_{1,+}^{(2)} := {} & \hat B \sin(F_1(\sigma)+\omega_1 \sigma)(\cos z^0(\sigma) - \cos z^*)\,,\\
f_{1,+}^{(3)} := {} & - \hat B p_z^0 (\sigma) \cos(F_1(\sigma)+\omega_1 \sigma)\,.
\end{align*}
Then, a straightforward computation shows that
\begin{align}
\left| \int_a^\infty f^{(1)}_{1,+}(\sigma) d\sigma \right| \leq {} & \hat B |p_y- \cos z^*| \left| \int_a^\infty (F_1(\sigma)-x_+) d\sigma \right|\,, \label{eq:M1,1} \\
\left| \int_a^\infty f^{(2)}_{1,+}(\sigma) d\sigma \right| \leq {} & \hat B 
\left| \int_a^\infty (4 \arctan \rme^{\lambda \sigma} - 2\pi)d\sigma \right|\,, \label{eq:M1,2} \\
\left| \int_a^\infty f^{(3)}_{1,+}(\sigma) d\sigma \right| \leq {} & \hat B 
\left| \int_a^\infty \frac{2\lambda}{\cosh(\lambda \sigma)} d\sigma \right|\,. \label{eq:M1,3}
\end{align}
By the asymptotic properties discussed in Section~\ref{ssec:geom:unper} we know that
these three integrals are convergent.
Next, we give some explicit expressions to control the above integrals. To this end, we use the 
expression of the primitives of the functions that we are integrating. First, we introduce
\[
g(\sigma):=\int (F_1(\sigma)-x_+) d\sigma = \frac{2}{\lambda^2} \log(\cosh(\lambda \sigma)) \sin z^* - \bigg( \frac{4}{\lambda^2} \arctan(\tanh(\tfrac{\lambda \sigma}{2}))-\frac{2\sigma}{\lambda}\bigg) \cos z^* - x_+ \sigma\,,
\]
which allows us to control Eq.~\eqref{eq:M1,1} in terms of the expression $g(\infty)-g(a)$. However, the direct evaluation of this expression with a computer presents
a huge rounding error. A more suitable formula is obtained using the limit
\[
\lim_{\sigma \rightarrow \infty} \log(\cosh(\lambda\sigma)) = 
\lim_{\sigma \rightarrow \infty} \log(\tfrac{\rme^{\lambda\sigma}+\rme^{-\lambda\sigma}}{2}) = \lim_{\sigma \rightarrow \infty} (\lambda \sigma - \log 2)\,,
\]
which allows us to control the term~\eqref{eq:M1,1} as follows
\[
\left| \int_a^\infty f^{(1)}_{1,+}(\sigma) d\sigma \right| \leq \hat B \frac{2}{\lambda^2} |p_y - \cos z^*|
\left|
 \log(1+\rme^{-2 \lambda a}) \sin z^*
- 2 \bigg(\arctan \Big(\frac{\rme^{-\lambda a}-1}{\rme^{-\lambda a}+1}\Big) + \frac{\pi}{4} \bigg) \cos z^*
\right|.
\]
Using Taylor series, it is easy to check that the term~\eqref{eq:M1,2} is bounded as
\[
\left| \int_a^\infty (4 \arctan \rme^{\lambda \sigma} - 2\pi)d\sigma \right|
\leq 4 \frac{\rme^{-\lambda a}}{\lambda} + \frac{4}{3} \frac{\rme^{-3 \lambda a}}{\lambda}\,,
\]
and that the term~\eqref{eq:M1,3} is estimated as
\[
\left| \int_a^\infty \frac{2\lambda}{\cosh(\lambda \sigma)} d\sigma \right| \leq 2 |\arctan(\sinh(\lambda a))-\pi|\,.
\]
Analogously, we can estimate the term 
$$\int_{-\infty}^{-a} f_{1,-}(\sigma)d\sigma\,,$$
thus proving the lemma.
\end{proof}

We use Lemma~\ref{lem:control:M} to rigorously control the Melnikov coefficients. Specifically, we evaluate directly the obtained expressions of $\Sigma_i$
using interval arithmetics. The integrals $\int_{-a}^0f_{i,-}$ and $\int_0^af_{i,+}$ are controlled using Simpson's rule
with rigorous bounds on the error, which are obtained using explicit
formulas for the 4th-order derivatives of the functions $f_i(\sigma)$
computed with a symbolic manipulator.

Next, we illustrate the rigorous evaluation of the Hypothesis $\mathbf{A}_1$ of Theorem~\ref{teo:diffusion:ABC}.
Obviously, it is enough to consider the coefficients $M_i^0$ for the parameters $\hat B=\hat C=1$
and store conveniently the obtained values.
If we are interested in other values of $\hat B$ and $\hat C$ we simply have to scale
the previously computed values.
In Table~\ref{tab:Melnikov} we present some rigorous enclosures of the coefficients $M_i^0$
corresponding to $\hat B= \hat C=1$, $p_x \in [0.4, 0.4001]$ and different
interval values of $p_y$. To control the tails
we use Lemma~\ref{lem:control:M}
with $a=20$ and to enclose the finite integrals we
use Simpson's rule with $130$ subintervals.
This implementation parameters are enough to guarantee 
that the coefficients $M_i^0$ 
do not vanish for a non-empty set $\cI$ of momenta. We can obtain a similar result
for a much larger domain $\cI$ by systematically performing this computation. 

\begin{table}[!h]
\centering
{\scriptsize
\begin{tabular}{|c| c c c c|}
\hline
$p_y$ & $M_1^0$ & $M_2^0$ & $M_3^0$ & $M_4^0$ \\
\hline \hline
$[0.1, 0.1001]$ & $[-5.1237, -4.9193]$ & $[-11.314, -10.972]$ & $[-1.282, -1.0754]$ & $[1.4611, 1.807]$ \\ 
$[0.2, 0.2001]$ & $[-5.3464, -5.1279]$ & $[-9.185, -8.8668]$ & $[-2.2718, -2.0471]$ & $[-3.0694, -2.7531]$ \\ 
$[0.3, 0.3001]$ & $[-5.911, -5.6582]$ & $[-6.5333, -6.2423]$ & $[-3.3557, -3.0967]$ & $[-5.5463, -5.2556]$ \\ 
$[0.4, 0.4001]$ & $[-6.4914, -6.2559]$ & $[-4.4201, -4.1843]$ & $[-4.4214, -4.1829]$ & $[-6.4917, -6.2557]$ \\ 
$[0.5, 0.5001]$ & $[-7.0075, -6.7721]$ & $[-2.9704, -2.7501]$ & $[-5.2892, -5.0544]$ & $[-6.7487, -6.5249]$ \\ 
$[0.6, 0.6001]$ & $[-7.4463, -7.2257]$ & $[-1.988, -1.7805]$ & $[-5.8843, -5.6657]$ & $[-6.701, -6.4904]$ \\ 
$[0.7, 0.7001]$ & $[-7.8251, -7.6315]$ & $[-1.3078, -1.129]$ & $[-6.2495, -6.0596]$ & $[-6.5182, -6.3291]$ \\ 
$[0.8, 0.8001]$ & $[-8.1753, -7.9784]$ & $[-0.85053, -0.66663]$ & $[-6.4659, -6.2811]$ & $[-6.2954, -6.0987]$ \\ 
$[0.9, 0.9001]$ & $[-8.4879, -8.2889]$ & $[-0.53277, -0.3423]$ & $[-6.575, -6.3929]$ & $[-6.0543, -5.8485]$ \\ 
\hline
\end{tabular}
\caption{{\footnotesize
We show rigorous enclosures of the coefficients $M_i^0$,
for the values $\hat B= \hat C=1$ and $p_x \in [0.4, 0.4001]$.}}
\label{tab:Melnikov}
}
\end{table}

Finally, let us illustrate how to check Hypothesis $\mathbf{A}_2$ of Theorem~\ref{teo:diffusion:ABC}. It turns out that,
for $p_x\in[0.4,0.4001]$ and all the intervals $p_y$ in Table~\ref{tab:Melnikov},
the four critical points of the map $(x,y) \mapsto \cL(x,y)$ in Eq.~\eqref{eq:crit:cL}
are non-degenerate and given by two maxima and two
saddle points, so that we can use the arguments in~\cite{DH11} to
see that there is a unique smooth critical point $\tau^*$ of the function
$\tau \mapsto \cL(x-\omega_1 \tau,y-\omega_2 \tau,p_x,p_y)$. In the
following section we discuss the rigorous enclosure of $\tau^*$ and
its derivatives.

\subsection{On the evaluation of the critical points}\label{ssec:cap:tau}

We discuss a simple methodology to rigorously
enclose the critical points of the function
\begin{align}
\tau \longmapsto \cL(x & -\omega_1 \tau, y-\omega_2 \tau,p_x,p_y) \nonumber \\
& = M_1^0 \cos(x-\omega_1 \tau) + M_2^0\cos(y-\omega_2 \tau)
+ M_3^0 \sin(x-\omega_1 \tau) + M_4^0\sin(y-\omega_2 \tau)\,, \label{eq:crit:poin}
\end{align}
where $(x,y,p_x,p_y) \in \TT^2 \times \cI$,
and the coefficients $M_i^0=M_i(0,p_x,p_y)$ are given by
Eqs.~\eqref{eq:M1:simple}--~\eqref{eq:M4:simple}.
The critical points $\tau^*=\tau^*(x,y,p_x,p_y)$ of~\eqref{eq:crit:poin} are characterized by the zeros
of the function
\[
Q(\tau):= \omega_1 \bigg(-M_1^0 \sin (x-\omega_1 \tau) + M_3^0 \cos(x-\omega_1 \tau)\bigg) + \omega_2 \bigg(-M_2^0 \sin(y-\omega_2 \tau) + M_4^0 \cos(y-\omega_2 \tau)\bigg)\,,
\]
which, of course, depends on the variables $(x,y,p_x,p_y)$. There are several techniques in computer-assisted proofs that allow us to study
solutions of nonlinear equations as above, e.g. the interval Newton method~\cite{Tucker11}. However, in the following discussion we choose to 
use the simplest possible method with the aim of convincing a reader that is not familiar
with computer-assisted methods. More advanced techniques
will give the possibility to validate large regions of phase-space with
a reduced computational cost. This is not the aim in this article, since the
required techniques are not related with the ideas that we want to highlight
and they would require to provide a larger amount of computational and implementation details.
To enclose $\tau^*$ we proceed using a bisection-like procedure:
\begin{itemize}
\item Given $(x,y,p_x,p_y)$, that may be numbers or interval values, we
enclose $M_i^0=M_i(0,p_x,p_y)$ following Section~\ref{ssec:cap:Mi}.

\item Given an integer $N>1$, we consider an increasing sequence
of interval values 
$\tau_i\in[(i-1)/N,i/N]$.
Then we compute
$Q(\tau_i)$. While $0 \notin Q(\tau_i)$ we increase the index $i$, until we obtain
an interval such that $0 \in Q(\tau_i)$. This gives a lower estimate for $\tau^*$.

\item After the previous computations, we continue the process of increasing the
index $i$ and computing $Q(\tau_i)$. When we reach an interval such that $0 \notin Q(\tau_i)$
then we have obtained an upper estimate for $\tau^*$.
\end{itemize}

As a result of the above procedure, we obtain
an interval enclosure of the critical point $\tau^*$.
Then, the derivatives of $\tau^*$ with respect to $(x,y)$
are computed using the following equations:
\[
\tau_\alpha^*     = - \frac{Q_\alpha(\tau^*)}{Q_\tau(\tau^*)}, \qquad
\tau_{\alpha\beta}^* = - \frac{Q_{\alpha\beta}(\tau^*)+Q_{\alpha\tau}(\tau^*) \tau^*_\beta + (Q_{\beta\tau}(\tau^*) + Q_{\tau\tau}(\tau^*) \tau^*_\beta) \tau^*_\alpha}{Q_\tau(\tau^*)}\,,
\]
where the subscripts denote, as usual, partial differentiation, and $\alpha$ and $\beta$ can be chosen to be $x$ or $y$.
For the case of the ABC system, we have
{\allowdisplaybreaks
\begin{align*}
Q_x (\tau) = {} & -M_1^0 \cos(x-\omega_1 \tau) \omega_1-M_3^0\sin(x-\omega_1 \tau) \omega_1\,, \\
Q_y (\tau)= {} & -M_2^0 \cos(y-\omega_2 \tau) \omega_2-M_4^0\sin(y-\omega_2 \tau) \omega_2\,, \\
Q_\tau(\tau)= {} & M_1^0\cos(x-\omega_1 \tau) \omega_1^2 + M_3^0 \sin(x-\omega_1 \tau) \omega_1^2 + M_2^0 \cos(y-\omega_2 \tau) \omega_2^2 + M_4^0 \sin(y-\omega_2 \tau) \omega_2^2\,, \\
Q_{xx} (\tau)= {} & M_1^0 \sin(x-\omega_1 \tau) \omega_1 - M_3^0 \cos(x-\omega_1 \tau) \omega_1\,, \\
Q_{yy} (\tau)= {} & M_2^0 \sin(y-\omega_2 \tau) \omega_2 - M_4^0 \cos(y-\omega_2 \tau) \omega_2\,, \\
Q_{x\tau}(\tau)={} & -M_1^0 \sin(x-\omega_1 \tau) \omega_1^2 - M_3^0 \cos(x-\omega_1 \tau) \omega_1^2\,, \\
Q_{y\tau}(\tau)={} & -M_2^0 \sin(y-\omega_2 \tau) \omega_2^2 - M_4^0 \cos(y-\omega_2 \tau) \omega_2^2\,, \\
Q_{\tau\tau}(\tau)={} & M_1^0 \sin(x-\omega_1 \tau) \omega_1^3 -M_3^0 \cos(x-\omega_1 \tau) \omega_1^3 + M_2^0 \sin(y-\omega_2 \tau) \omega_2^3- M_4^0 \sin(y-\omega_2 \tau) \omega_2^3\,.
\end{align*}}

In Table~\ref{tab:critical:points} we illustrate these computations taking $\hat B = \hat C=1$ and considering the values of $p_x$ and $p_y$
used in Table~\ref{tab:Melnikov}. In all the cases, we fix the angles as $x=y=0$, and we
compute the critical point $\tau^*$ using $N=100$.

\begin{table}[!h]
\centering
{\scriptsize
\begin{tabular}{|c| c c c|}
\hline
$p_y$ & $\tau^*$ & $\tau^*_x$ & $\tau^*_y$ \\
\hline \hline
$[0.1, 0.1001]$ & $[2.34, 2.43]$ & $[0.74112, 0.91778]$ & $[-0.41917, -0.3401]$ \\ 
$[0.2, 0.2001]$ & $[2.71, 2.81]$ & $[0.47369, 0.73294]$ & $[0.29594, 0.45414]$ \\ 
$[0.3, 0.3001]$ & $[2.38, 2.45]$ & $[0.43611, 0.54436]$ & $[0.41137, 0.51753]$ \\ 
$[0.4, 0.4001]$ & $[2.1, 2.16]$ & $[0.40377, 0.50552]$ & $[0.40346, 0.5051]$ \\ 
$[0.5, 0.5001]$ & $[1.9, 1.96]$ & $[0.37234, 0.4717]$ & $[0.39116, 0.50717]$ \\ 
$[0.6, 0.6001]$ & $[1.75, 1.81]$ & $[0.33403, 0.42893]$ & $[0.38246, 0.51976]$ \\ 
$[0.7, 0.7001]$ & $[1.61, 1.67]$ & $[0.27007, 0.36171]$ & $[0.38473, 0.54755]$ \\ 
$[0.8, 0.8001]$ & $[1.48, 1.55]$ & $[0.1825, 0.2837]$ & $[0.37799, 0.60038]$ \\ 
$[0.9, 0.9001]$ & $[1.36, 1.44]$ & $[0.092554, 0.18709]$ & $[0.37117, 0.65404]$ \\ 
\hline \hline
$p_y$ & $\tau^*_{xx}$ & $\tau^*_{xy}$ & $\tau^*_{yy}$ \\
\hline \hline
$[0.1, 0.1001]$ & $[0.45478, 1.4543]$ & $[-0.25048, 0.00020579]$ & $[-0.40322, -0.24827]$ \\ 
$[0.2, 0.2001]$ & $[0.48773, 1.5132]$ & $[0.22382, 0.81026]$ & $[-0.363, 0.057794]$ \\ 
$[0.3, 0.3001]$ & $[0.41841, 0.76705]$ & $[0.32713, 0.5756]$ & $[0.18447, 0.45382]$ \\ 
$[0.4, 0.4001]$ & $[0.24271, 0.51779]$ & $[0.34875, 0.58267]$ & $[0.40865, 0.70313]$ \\ 
$[0.5, 0.5001]$ & $[0.066261, 0.29543]$ & $[0.36568, 0.62845]$ & $[0.62318, 1.0275]$ \\ 
$[0.6, 0.6001]$ & $[-0.11457, 0.079559]$ & $[0.38581, 0.69138]$ & $[0.85943, 1.4404]$ \\ 
$[0.7, 0.7001]$ & $[-0.32399, -0.12284]$ & $[0.37369, 0.72931]$ & $[1.0946, 1.9373]$ \\ 
$[0.8, 0.8001]$ & $[-0.53479, -0.27236]$ & $[0.29473, 0.75013]$ & $[1.2426, 2.6407]$ \\ 
$[0.9, 0.9001]$ & $[-0.68739, -0.39066]$ & $[0.18037, 0.65021]$ & $[1.3104, 3.4064]$ \\ 
\hline
\end{tabular}
\caption{{\footnotesize
We show rigorous enclosures of the critical point $\tau^*$
and its derivatives with respect to $(x,y)$
for the values $B=C=1$, $p_x \in [0.4, 0.4001]$, and
$x=y=0$.}}
\label{tab:critical:points}
}
\end{table}

\subsection{On the verification of the transversality conditions}\label{ssec:cap:delta}

With the rigorous estimates obtained in Sections~\ref{ssec:cap:Mi} and~\ref{ssec:cap:tau},
let us now explain how to use interval arithmetics to check the transversality condition in Hypothesis $\mathbf{A}_3$
of Theorem~\ref{teo:diffusion:ABC}. This consists in enclosing the functions
$\{\Delta_i\}_{i=1,2,3}$ and $\{\hat \Delta_i\}_{i=1,2,3,4}$
given by Eqs.~\eqref{eq:Delta:1}--~\eqref{eq:Delta:7}.
In order to check the condition in the resonant region, we observe that $\hat \Delta_3$
and $\hat \Delta_4$ are proportional to $p_x-p_y$ so that they tend to zero when we
approach the resonance. For this reason, we eliminate the factor $p_x-p_y$ in the computation
of the expression  $\hat \Delta_1 \hat \Delta_4 -\hat \Delta_2 \hat \Delta_3$. This
allows us to check that the condition holds in any tubular neighborhood that is close enough to the resonance.

An important observation it that we have the freedom of choosing the angles $(x,y)$
to evaluate the critical point~$\tau^*$. In fact, there is no
optimal way to choose the angles $(x,y)$ in order to verify the transversality conditions.
The reason is that optimal values selected numerically may fail to fulfill such conditions when rigorous interval operations are used. This is because the enclosed value of $Q_\tau(\tau^*)$
may be very close to zero (or even contain this point), thus producing a large enclosure
in the evaluation of the conditions. Our experience in this problem is that choosing random values of $(x,y)$,
until we reach a suitable pair,
is the simplest and fast strategy.

For example, in Table~\ref{tab:transv} we present some rigorous enclosures of the
transversality conditions in Hypothesis $\mathbf{A}_3$
of Theorem~\ref{teo:diffusion:ABC},
corresponding to $\hat B= \hat C=1$, $p_x \in [0.4, 0.40001]$, and different
interval values of $p_y$. To control the tails
we use Lemma~\ref{lem:control:M}
with $a=20$ and to enclose the finite integrals we
use Simpson's rule with $300$ subintervals. To obtain the
critical point and its derivatives, we use the approach described
in Section~\ref{ssec:cap:tau} with $N=300$.
This implementation parameters are enough to guarantee 
that the functions $\Delta_1 \Delta_3 - \Delta_2^2$
and $\hat \Delta_1 \hat \Delta_4 -\hat \Delta_2 \hat \Delta_3$
do not vanish for a non-empty set $\cI$ of momenta. By computing
simultaneously the condition in the non-resonant region (4th column
of Table~\ref{tab:transv}) and
in the resonant region (5th column
of Table~\ref{tab:transv}), it is clear that we can select a number $L>0$
that allows us to obtain diffusing orbits crossing the resonance.

\begin{table}[!h]
\centering
{\scriptsize
\begin{tabular}{|c| c c c c|}
\hline
$p_y$ & $x$ & $y$ & $\Delta_1 \Delta_3 - \Delta_2^2$ & $\frac{\hat \Delta_1 \hat \Delta_4 -\hat \Delta_2 \hat \Delta_3}{p_x-p_y}$ \\
\hline \hline
$[0.1, 0.10001]$ & $5.2923$ & $0.93117$ & $[-26.899, -3.5905]$ & $[2.3401, 101.72]$\\
$[0.2, 0.20001]$ & $2.7665$ & $0.55732$ & $[8.3643, 23.588]$   & $[-68.148, -6.1705]$\\
$[0.3, 0.30001]$ & $1.1969$ & $0.37322$ & $[26.796, 48.326]$   & $[-135.08, -61.886]$\\
$[0.4, 0.40001]$ & $3.7869$ & $4.19530$ & $[11.818, 28.517]$   & $[-73.026, -30.082]$\\
$[0.5, 0.50001]$ & $4.9160$ & $0.61701$ & $[-13.819, -6.4592]$ & $[18.387, 34.8]$\\
$[0.6, 0.60001]$ & $1.7342$ & $2.53840$ & $[-17.72, -3.2241]$  & $[8.7055, 41.964]$\\
$[0.7, 0.70001]$ & $4.7928$ & $5.74470$ & $[-16.15, -6.0621]$  & $[13.441, 36.899]$\\
$[0.8, 0.80001]$ & $5.0542$ & $6.19650$ & $[-20.256, -6.1962]$ & $[12.115, 45.704]$\\
$[0.9, 0.90001]$ & $1.5249$ & $4.18960$ & $[-6.1939, -0.93573]$& $[0.38193, 12.601]$\\
\hline
\end{tabular}
\caption{{\footnotesize
We show rigorous enclosures of the transversality conditions of Theorem~\ref{teo:diffusion:ABC}
for the values $\hat B= \hat C=1$, $p_x \in [0.4, 0.40001]$, for different interval values of $p_y$. The critical point $\tau^*$ and its derivatives are evaluated at the points $(x,y)$
shown in the 2nd and the 3rd columns.}}
\label{tab:transv}
}
\end{table}

Finally, we describe the implementation parameters of the CAP that lead to
the result stated in Corollary~\ref{cor:diffusion:ABC:informal}. We
take $\hat B=10$ and $\hat C=0.1$ and we
divide the set $\cI=[0.1,0.9] \times [0.5,0.9]$ in subsets
of size $10^{-4} \times 10^{-4}$. For every subset, we
use Lemma~\ref{lem:control:M}
with $a=20$ and to enclose the finite integrals we
use Simpson's rule with $130$ subintervals. To obtain the
critical point and its derivatives, we use the approach described
in Section~\ref{ssec:cap:tau} with $N=100$. For all these sets
we obtain that the function $\Delta_1 \Delta_3 - \Delta_2^2$
does not vanish in $\cI$, and the function
$\hat \Delta_1 \hat \Delta_4 -\hat \Delta_2 \hat \Delta_3$
only vanishes on the resonant line $p_x=p_y$.

\section*{Acknowledgements}

The authors are very grateful to A. Delshams, M. Guardia, A. Haro, G. Huguet, R. de la Llave, and T.M. Seara for useful discussions and suggestions. We especially want to thank T.M. Seara for her patience and kindness answering several questions on the papers~\cite{DLS06,DLS08,DLS13}. The authors are supported by the ERC Starting Grant~335079. This work is supported in part by the ICMAT--Severo Ochoa grant SEV-2011-0087 and the grants
MTM2012-3254 (A.L.) and 2014SGR1145 (A.L.).

\end{document}